\documentclass[twocolumn,10pt,superscriptaddress,amsmath,amssymb,aps,pre,floatfix,]{revtex4-2}

\usepackage{graphicx}
\usepackage{bm}
\usepackage{times}
\usepackage{amsmath,amssymb}
\usepackage{amsthm}
\usepackage{enumerate}
\usepackage{multirow}
\usepackage{siunitx}
\usepackage{makecell}
\usepackage{physics}
\usepackage{comment}
\usepackage[caption=false]{subfig}
\usepackage{tikz}
\usetikzlibrary{quantikz2}
\usepackage{tabularx}

\usepackage[colorlinks=true,linkcolor=blue,citecolor=blue,urlcolor=blue,linktocpage=true,breaklinks=true]{hyperref}

\newtheorem{theorem}{Theorem} 
\newtheorem{lemma}[theorem]{Lemma} 
\newtheorem{corollary}[theorem]{Corollary}

\newtheorem{proposition}[theorem]{Proposition}

\begin{document}

\title{Negative quasiprobability trajectories for Bell-diagonal states under local decoherence}

\author{Qing \surname{Yu}}
\email{qingyu2003@stumail.nwu.edu.cn}
\affiliation{School of Physics, Northwest University, Xi’an 710127, China}

\author{Kun \surname{Zhang}}
\email{kunzhang@nwu.edu.cn}
\affiliation{School of Physics, Northwest University, Xi’an 710127, China}
\affiliation{Shaanxi Key Laboratory for Theoretical Physics Frontiers, Xi'an 710127, China}
\affiliation{Peng Huanwu Center for Fundamental Theory, Xi'an 710127, China}
\affiliation{Fundamental Discipline Research Center for  Quantum Science and Technology of Shaanxi Province, Xi'an 710127, China}

\author{Yang-Yang Chen}
\email{chenyy@nwu.edu.cn}
\affiliation{Institute of Modern Physics, Northwest University, Xi'an 710127, China}
\affiliation{Shaanxi Key Laboratory for Theoretical Physics Frontiers, Xi'an 710127, China}
\affiliation{Peng Huanwu Center for Fundamental Theory, Xi'an 710127, China}
\affiliation{Fundamental Discipline Research Center for Quantum Science and Technology of Shaanxi Province, Xi'an 710127, China}

\author{Xiao-Hui Wang}
\email{xhwang@nwu.edu.cn}
\affiliation{School of Physics, Northwest University, Xi’an 710127, China}
\affiliation{Shaanxi Key Laboratory for Theoretical Physics Frontiers, Xi'an 710127, China}
\affiliation{Peng Huanwu Center for Fundamental Theory, Xi'an 710127, China}
\affiliation{Fundamental Discipline Research Center for Quantum Science and Technology of Shaanxi Province, Xi'an 710127, China}

\date{\today}

\begin{abstract}
The fluctuation theorem (FT) relates microscopic trajectory distributions to macroscopic averages. Using quasiprobability trajectories, this framework has recently been extended to quantum information dynamics. Despite sharing the form of conventional thermodynamic FTs, information FTs involve quasiprobabilities whose negativity and statistical properties remain poorly understood. In this work, we analytically determine the properties of negative distributions for two-qubit Bell-diagonal states evolving under local dephasing, depolarizing, and amplitude-damping channels, under the fixed local computational basis. We show that the occurrence of negative quasiprobabilities is not determined by the entanglement or Bell nonlocality of the initial state. For these channels, a diagonal-interference decomposition identifies a nonnegative reference contribution and yields channel-dependent interference bounds. More importantly, we prove that quasiprobability distributions containing negative weights can violate Horv\'ath's necessary and sufficient criterion for the Jensen-Steffensen inequality to hold for every continuous convex function, even though the specific exponential Jensen-like relation enforced by the FT and the quantum data-processing inequality remains valid. This trajectory-level contrast with classical stochastic descriptions reveals a distinctive feature of quantum information dynamics that is invisible when only the corresponding average-level information inequalities are considered. 
\end{abstract}

\maketitle

\section{\label{sec:intro} Introduction}

The fluctuation theorem (FT) provides a microscopic refinement of the macroscopic second law by constraining the full probability distribution of trajectory-level thermodynamic quantities rather than only their averages~\cite{Evans1993,Jarzynski1997,Crooks1999,Seifert2005}. It relates the probability ratio of forward and time-reversed microscopic trajectories to entropy production, work, or heat, thereby quantifying irreversibility far from equilibrium. More generally, integral FTs can be understood from a retrodictive viewpoint in which normalization of an appropriately defined reverse process plays a central role~\cite{BuscemiScarani2021,AwBuscemiScarani2021}. FTs have consequently become a central framework of stochastic and quantum thermodynamics~\cite{Esposito2009,Campisi2011,FunoUedaSagawa2018,LandiPaternostro2021}.

Information modifies the thermodynamic structure when a system exchanges information with a memory, observer, or feedback controller. In such settings, information terms enter the FT together with entropy production and lead to information-corrected second laws~\cite{SagawaUeda2010,HorowitzVaikuntanathan2010,SagawaUeda2012,HorowitzEsposito2014}. A more intriguing generalization is to directly characterize information dynamics within the FT framework, rather than treating information as a correction. Recently, information quantities such as correlations, coherence, and out-of-time-ordered correlators have been shown to satisfy quantum FT identities~\cite{Halpern2016JarzynskilikeEF,Halpern2017QuasiprobabilityBT,Micadei2020,Zhang2021ConditionalEP,Zhang2022QuasiprobabilityFT,Zhang2026Multipartite}. Rather than implying a thermodynamic inequality, an informational FT corresponds to an information inequality such as the quantum data-processing inequality~\cite{LiebRuskai1973,SchumacherNielsen1996,buscemi2014complete}.

An important distinction separates information FTs from conventional quantum thermodynamic FTs. The latter are commonly formulated using the two-point-measurement scheme, in which thermodynamic changes are assigned to pairs of projective-measurement outcomes and the intervening quantum evolution enters through conditional transition probabilities~\cite{Tasaki2000,Talkner2007}. Such a classical trajectory description generally cannot be constructed for quantum information dynamics. Moreover, an initial projective measurement can erase the coherence and correlations whose dynamics the information FT is intended to characterize. Trajectories given by the Kirkwood-Dirac (KD) quasiprobability resolve this obstruction by assigning weights to sequences of incompatible operators~\cite{kirkwood1933quantum,dirac1945rmp,lostaglio2023kd,ArvidssonShukur2024KD,GherardiniDiChiara2024}. They preserve the normalized trajectory description, but their weights may be negative or complex and hence differ fundamentally from classical probability distributions.

Negative and nonreal KD quasiprobabilities have been studied in connection with contextuality, incompatibility, scrambling, and quantum advantages~\cite{Lostaglio2018,gonzalez2019out,Arvidsson-Shukur_2021,de2021kirkwood,Pashayan2015,LevyLostaglio2020}. However, their quantitative properties and statistical distributions have rarely been explored. For example, how negative can a quasiprobability become in the context of quantum information FTs? How is the negativity distributed among quasiprobability trajectories? In addition, there is a tension between the standard Jensen's inequality and the Jensen-Steffensen inequality, since the latter admits negative weights under suitable conditions. 

In this work, we use the analytically tractable decoherence dynamics of two-qubit Bell-diagonal states under local dephasing, depolarizing, and amplitude-damping channels to investigate the negativity and statistical distribution of the KD quasiprobability trajectories associated with the FT for quantum mutual-information dissipation~\cite{Zhang2026Multipartite}. Throughout this channel comparison, we use a fixed computational-basis rank-one projector of the local states. A diagonal-interference decomposition yields channel-dependent bounds on the contributions responsible for negative trajectory weights. For amplitude damping, we analytically determine the most negative quasiprobability trajectory weight, together with the parameter-independent bounds on its negativity. Unexpectedly, the amplitude-damping examples show that the resulting KD quasiprobability trajectory distributions can violate Horv\'ath's necessary and sufficient criterion for the Jensen-Steffensen inequality to hold for every continuous convex function~\cite{Horvath2024}. At the microscopic trajectory level, this finding shows that classical and quantum information inequalities can coexist at the average level while resting on different stochastic structures. Here, the quantum relation is protected by the logarithmic information variable and its associated exponential FT, rather than by a nonnegative trajectory distribution that supports Jensen's inequality for arbitrary convex functions. 

The paper is organized as follows. Section~\ref{sec:quasi_FT} reviews the quasiprobability FT for mutual information dissipation and derives its operator-sum trajectory representation. Section~\ref{sec:dephase_depolarizing} applies this framework to Bell-diagonal states under local dephasing and depolarizing channels. Section~\ref{sec:decay} analyzes local amplitude damping, identifies the negative trajectories, and determines the most negative trajectory weight. Section~\ref{sec:neagative} develops the interference decomposition and channel-dependent interference bounds. Section~\ref{sec:violation_horvath} examines the validity of Horv\'ath's criterion based on the amplitude damping dynamics. Section~\ref{sec:conclusion} summarizes the results and discusses future directions. The appendices provide the detailed calculations and proofs underlying the results in the main text.

\section{\label{sec:quasi_FT} Quasiprobability trajectory and FT}

Section~\ref{sub:QFT} reviews the quasiprobability FT for mutual information dissipation. In Sec.~\ref{subsec:kraus_quasiprobability}, we then express the associated trajectory quasiprobability in Kraus-operator form to facilitate the subsequent calculations.

\subsection{\label{sub:QFT}Quasiprobability FT for mutual information dissipation}
Consider a bipartite system labeled $S=S_1S_2$ with an initial state denoted by $\rho_S$. The dissipative dynamics is modeled by allowing each qubit to interact locally with its own environment, and the global evolution is therefore described by $U_{SE}=U_{S_1E_1}\otimes U_{S_2E_2}$. The total environment is initially prepared in a product state $\rho_E=\rho_{E_1}\otimes \rho_{E_2}$. The final joint state is
\begin{equation}
\rho'_{SE}=U_{SE}(\rho_S\otimes\rho_E)U_{SE}^\dagger,
\end{equation}
and tracing out the environments yields the final system state $\rho'_S=\Tr_E(\rho'_{SE})$. In what follows, we use the convention that variables associated with the final state are denoted by a prime, while initial-state variables are unprimed.

The total correlation between the two qubits is quantified by the quantum mutual information~\cite{wilde2017quantum}
\begin{equation}
\mathcal{I}(\rho_S)= \sum_{j=1}^2 S(\rho_{S_j}) - S(\rho_S),
\end{equation}
where $S(\rho)=-\Tr(\rho\ln\rho)$ is the von Neumann entropy. The mutual information change is therefore $\Delta \mathcal{I}=\mathcal{I}(\rho_S)-\mathcal{I}(\rho'_S)$. Explicitly, we have
\begin{equation}
\label{eq:Delta_I}
\Delta \mathcal{I} = \sum_{j=1}^2 S(\rho_{S_j}) - S(\rho_{S}) - \sum_{j=1}^2 S(\rho'_{S_j}) + S(\rho'_S).
\end{equation}
Since the evolution acts locally on $S_1$ and $S_2$, the mutual information cannot increase, so $\Delta \mathcal{I}\geq 0$ by the quantum data-processing inequality~\cite{LiebRuskai1973,SchumacherNielsen1996,buscemi2014complete}. Note that the bipartite mutual information admits a direct multipartite generalization involving the von Neumann entropies of all subsystems, which is also called the total quantum correlation~\cite{DeChiara2017GenuineQC}.

For a density matrix $\rho$ with eigenvalues $p_k$, the von Neumann entropy can be expressed as $S(\rho)=-\sum_k p_k\ln p_k$. Therefore, the quantity $-\ln p_k$ can be viewed as a stochastic entropy associated with the distribution $p_k$~\cite{Seifert2005}. Following the same logic as in stochastic thermodynamics, we can define a stochastic mutual information change, denoted by
\begin{equation}
\label{eq:stochastic-mutual information}
\Delta \iota[\gamma]= -\sum_{j=1}^2 \ln p_{s_j} + \ln p_k + \sum_{j=1}^2 \ln p_{s'_j} - \ln p_{k'},
\end{equation}
where $\gamma=\{k,s_1,s_2,k',s'_1,s'_2\}$ collects the stochastic variables, and the probabilities $p_k$ and $p_{s_j}$ are the eigenvalues of the global state $\rho_S$ and the local states $\rho_{S_j}$, respectively. Formally, we have $\rho_S=\sum_k p_k \Pi_k$ and $\rho_{S_j}=\sum_{s_j}p_{s_j}\Pi_{s_j}$ with $j=1,2$, where $\Pi_k$ and $\Pi_{s_j}$ are rank-one eigenprojectors. For a degenerate eigenvalue, the projector onto the full eigenspace is unique, but its rank-one projector must be specified separately as part of the trajectory definition. The same convention applies to the final states. The probabilities associated with the final states can be defined accordingly using primed variables. Note that we have a term-by-term, one-to-one correspondence between the mutual information change $\Delta\mathcal I$ in Eq.~\eqref{eq:Delta_I} and the stochastic mutual information change $\Delta\iota[\gamma]$ in Eq.~\eqref{eq:stochastic-mutual information}.

The bridge between the stochastic mutual information change and the average mutual information change is provided by a stochastic description of the dynamics, which contains the correct marginal distributions for the global and local states at both the initial and final times. However, the spectral projectors of the global and local states need not commute, and therefore a classical joint distribution of the stochastic variables $\gamma$ cannot generally be assumed. Instead, a quasiprobability distribution can reproduce the correct marginal distributions for noncommuting observables~\cite{lostaglio2023kd,ArvidssonShukur2024KD,GherardiniDiChiara2024}. Specifically, we consider the KD quasiprobability distribution defined as follows~\cite{Zhang2026Multipartite}
\begin{equation}
\label{eq:quantum-trajectory-quasiprobability}
\mathcal{Q}[\tilde\gamma]
=\Tr\!\left(
U_{SE}^\dagger \Pi_{k'}\Pi_{s'}\Pi_{n'} U_{SE}\Pi_s\Pi_n\Pi_k \,\rho_S\rho_E
\right),
\end{equation}
with $\tilde\gamma = \gamma \cup \{n_1,n_2,n'_1,n'_2\}$, where $n_j$ and $n'_j$ label the stochastic variables associated with the initial and final environmental states, respectively, whose spectral decompositions are $\rho_{E_j} = \sum_{n_j} p_{n_j} \Pi_{n_j}$ and $\rho'_{E_j} = \sum_{n'_j} p_{n'_j} \Pi_{n'_j}$. We simplify the notation as $\Pi_s=\Pi_{s_1}\otimes\Pi_{s_2}$ and $\Pi_{n}=\Pi_{n_1}\otimes\Pi_{n_2}$. The tensor-product symbol has been omitted in Eq. \eqref{eq:quantum-trajectory-quasiprobability} for brevity and without ambiguity. The KD quasiprobability distribution $\mathcal{Q}[\tilde\gamma]$ is normalized, i.e., $\sum_{\tilde\gamma} \mathcal{Q}[\tilde\gamma] = 1$. Its marginals reproduce the correct distributions of the initial and final global and local states. However, its values can be negative or complex for noncommuting projectors, which is a signature of quantum interference.

The KD quasiprobability distribution $\mathcal Q[\tilde \gamma]$ correctly reproduces the average mutual information change $\Delta \mathcal{I}$. Formally, we have
\begin{equation}
\label{eq:classical-mutual information}
\langle\Delta\iota\rangle_{\mathcal{Q}}=\sum_{\tilde\gamma}
\mathcal{Q}[\tilde\gamma]\,\Delta\iota[\gamma]=\Delta\mathcal{I}.
\end{equation}
Moreover, the exponential average of the stochastic mutual information change satisfies the integral FT~\cite{Zhang2026Multipartite}
\begin{equation} 
\label{eq:FT}
\langle e^{-\Delta\iota}\rangle_{\mathcal{Q}}=\sum_{\tilde\gamma}
\mathcal{Q}[\tilde\gamma]\,e^{-\Delta\iota[\gamma]}=1.
\end{equation}
Although the mathematical form of the above integral FT appears identical to that of the standard FT~\cite{Campisi2011}, it differs from the traditional quantum FT in two respects. First, it describes the statistics of a purely informational quantity associated with an informational inequality $\Delta\mathcal I\geq 0$. Second, it employs a quasiprobability distribution that admits negative and complex values.

\subsection{\label{subsec:kraus_quasiprobability}Operator-sum representation of the quasiprobability trajectory}

There is a mismatch between the stochastic variables of $\Delta\iota[\gamma]$ and $\mathcal Q[\tilde\gamma]$, since the latter contains the environmental variables $n$ and $n'$. As we are concerned only with the statistics of $\Delta\iota[\gamma]$, we sum over the environmental variables of $\mathcal Q[\tilde\gamma]$, which is analogous to taking a partial trace over the density matrix and the unitary evolution operator \cite{NielsenChuang2010}. In what follows, we therefore derive the operator-sum representation of the quasiprobability trajectory of $\mathcal Q[\tilde\gamma]$. 

For each subsystem $S_j$ with $j=1,2$ locally interacting with its environment $E_j$, the dynamics of the subsystem alone can be described by a general quantum operation, also known as the Kraus operator-sum representation. Specifically, we have
\begin{equation}
\rho'_{S_j} = \mathcal E_{S_j}(\rho_{S_j}) =\sum_{l_j,n_j}K_{l_j,n_j}\rho_{S_j}
K^{\dagger}_{l_j,n_j},
\end{equation}
where $K_{l_j,n_j}$ is the Kraus operator given by $K_{l_j,n_j}
=\sqrt{p_{n_j}}\,\bra{e_{l_j}}U_{S_jE_j}\ket{n_j}$. Here, $j=1,2$ label the subsystems, and $\{\ket{e_{l_j}}\}$ is an orthonormal basis of the environment $E_j$. The Kraus operators satisfy the completeness relation $\sum_{l_j,n_j}K^{\dagger}_{l_j,n_j}K_{l_j,n_j}=\mathbb{I}_{S_j}$, which ensures that the map $\mathcal E_{S_j}$ is trace preserving. Because the two subsystem-environment interactions are local, the reduced dynamics of the bipartite system is governed by the product channel $\rho'_S=\mathcal E_S(\rho_S)$ with $\mathcal E_S=\mathcal E_{S_1}\otimes\mathcal E_{S_2}$. 

Based on the above Kraus operator-sum representation, we now eliminate the environmental trajectory variables of the KD quasiprobability $\mathcal Q[\tilde\gamma]$. With $n=\{n_1,n_2
\}$ and $n'=\{n'_1,n'_2\}$, we define the reduced system KD quasiprobability as
\begin{equation}
\mathcal Q_S[\gamma] =\sum_{n,n'}\mathcal Q[\tilde\gamma].
\end{equation}
Because of the completeness relation of the projectors $\Pi_n$ and $\Pi_{n'}$, we obtain
\begin{multline}
\label{eq:Q_S}
\mathcal Q_S[\gamma] = \Tr_{SE}\left(U_{SE}^{\dagger}\Pi_{k'}\Pi_{s'}U_{SE}\Pi_s\Pi_k\rho_S\rho_E\right) \\
= \Tr_S\left(
\Pi_{k'}\Pi_{s'}\mathcal E_S(\Pi_s\Pi_k\rho_S)\right).
\end{multline}
Using $\Pi_k\rho_S=p_k\Pi_k$, which follows from the spectral decomposition of $\rho_S$, the KD quasiprobability can be expressed as
\begin{equation}
\label{eq:QS-q-pk}
\mathcal Q_S[\gamma]
=p_k\,q[\gamma],
\end{equation}
where $q[\gamma]$ serves as the transition quasiprobability of the dynamics, given by 
\begin{equation}
q[\gamma] = \Tr\left(
\Pi_{k'}\Pi_{s'}\mathcal E_S(\Pi_s\Pi_k)\right).
\end{equation}
For rank-one projectors, the transition quasiprobability $q[\gamma]$ above takes the amplitude form
\begin{equation}
q[\gamma]
=\langle k'|s'\rangle
\bra{s'}
\mathcal E_S(\Pi_s\Pi_k)
\ket{k'}.
\label{eq:q-operator-sum}
\end{equation}
The reduced KD quasiprobability $\mathcal Q_S[\gamma]$ is properly normalized, $\sum_{\gamma}\mathcal Q_S[\gamma]=1$. By contrast, $q[\gamma]$ is a conditional transition quasiprobability, which satisfies $\sum_{s,k',s'}q[k,s,k',s']=1$ for every fixed initial global label $k$. For $p_k>0$, negativity or nonreality of $q[\gamma]$ is inherited by $\mathcal Q_S[\gamma]=p_kq[\gamma]$. When $p_k=0$, the corresponding reduced trajectory weight vanishes.

Correspondingly, the first moment of $\Delta\iota[\gamma]$ and the integral FT in Eqs.~\eqref{eq:classical-mutual information} and~\eqref{eq:FT} can be expressed in terms of the reduced KD quasiprobability $\mathcal Q_S[\gamma]$ as
\begin{subequations} 
\begin{align}
\label{eq:Qs_avg}
&\langle\Delta\iota\rangle_{\mathcal Q_S}
=\sum_{\gamma}\mathcal Q_S[\gamma]\,
\Delta\iota[\gamma]
=\Delta\mathcal I,\\
\label{eq:Qs_fluct}
&\langle e^{-\Delta\iota}\rangle_{\mathcal Q_S}
=\sum_{\gamma}\mathcal Q_S[\gamma]\,
e^{-\Delta\iota[\gamma]}
=1.
\end{align}
\end{subequations}
The above quasiprobability FT resembles the quasiprobability FT for a quantum channel \cite{KwonKim2019}. However, the main difference is that the quasiprobability FT here concerns mutual information dissipation, whereas the latter concerns a complex entropy production defined by the ratio of the forward and backward KD quasiprobability trajectories.

We evaluate logarithmic trajectory statistics on the nonzero support of $\mathcal Q_S$. The unit FT assumes strictly positive initial and final spectra. In conventional FTs, zero-probability trajectories can give rise to absolute irreversibility, so that the exponential sum need not equal unity~\cite{Funo2015AbsoluteIrreversibility}. Extending quasiprobability FTs to such cases requires separate study and lies beyond the scope of this work.

\section{\label{sec:dephase_depolarizing}Quasiprobability distributions under dephasing and depolarizing channels}

In this section, we first introduce the basic properties of Bell-diagonal states in Sec.~\ref{Bell diagonal state}. We then analytically examine the KD quasiprobability distribution of Bell-diagonal states under local dephasing and depolarizing dynamics in Sec.~\ref{dephasing channel} and Sec.~\ref{depolasirprobability}, respectively. We show that all reduced trajectory weights remain nonnegative, with a fixed computational-basis local projector, for both channels irrespective of the initial state parameters. 

\subsection{\label{Bell diagonal state}Bell-diagonal states}
We consider a two-qubit Bell-diagonal state of the form~\cite{Horodecki1995}
\begin{equation}
\rho_S=\frac{1}{4}\left(
\mathbb{I}_{S_1}\otimes\mathbb{I}_{S_2}
+\sum_{\alpha=1}^3c_\alpha\,
\sigma^\alpha_{S_1}\otimes\sigma^\alpha_{S_2}
\right),
\end{equation}
where $\sigma^\alpha$ ($\alpha=1,2,3$) are the three Pauli operators and $c_\alpha\in\mathbb{R}$ are the correlation coefficients. Positivity of $\rho_S$ requires
\begin{subequations}
\label{eq:Bell_diagonal_c_inequality}
\begin{align}
1-c_1-c_2-c_3&\geq0,\\
1-c_1+c_2+c_3&\geq0,\\
1+c_1-c_2+c_3&\geq0,\\
1+c_1+c_2-c_3&\geq0.
\end{align}
\end{subequations}
These four inequalities are equivalent to the nonnegativity of the four Bell-state populations given below. Bell-diagonal states are so named because they are diagonal in the Bell basis, whose vectors are
\begin{subequations}
\label{eq:Bell-states}
\begin{align}
\ket{\Phi^\pm}&=\frac{1}{\sqrt{2}}
\left(\ket{00}\pm\ket{11}\right),\\
\ket{\Psi^\pm}&=\frac{1}{\sqrt{2}}
\left(\ket{01}\pm\ket{10}\right).
\end{align}
\end{subequations}
Equivalently, the four Bell states are the eigenstates of $\rho_S$. We denote the eigenstate labels by $k\in\{\Phi^\pm,\Psi^\pm\}$. Introducing the variables $u=c_1+c_2$ and $v=c_1-c_2$, Bell-diagonal states have the corresponding eigenvalues
\begin{equation}
\label{eq:pk}
p_{k=\Phi^\pm}=\frac{1}{4}\left(1+c_3\pm v\right),\quad 
p_{k=\Psi^\pm}=\frac{1}{4}\left(1-c_3\pm u\right).
\end{equation}
Since Bell-diagonal states are classical mixtures of the four Bell states, their partial traces give the maximally mixed local states, namely $\rho_{S_1}=\rho_{S_2}=\mathbb{I}/2$. The eigenvalues of the local states are all $p_{s_j=0}=p_{s_j=1}=1/2$. The degenerate local eigenspaces admit different rank-one refinements. We choose the computational-basis projectors $\Pi_{s_j}=\ket{s_j}\!\bra{s_j}$ with $s_j\in\{0,1\}$ at both initial and final times. Correspondingly, the joint local state projectors are denoted as $\Pi_s=\Pi_{s_1}\otimes\Pi_{s_2}=\ket{s}\!\bra{s}$ with $s\in\{00,01,10,11\}$. 

Bell-diagonal states form one of the simplest two-qubit state families while retaining a rich structure of quantum correlations and permitting a fully analytic treatment. The positivity conditions in Eq.~\eqref{eq:Bell_diagonal_c_inequality} define a tetrahedron in $(c_1,c_2,c_3)$ space, whose vertices correspond to the four pure Bell states. The inscribed octahedron, defined by $|c_1|+|c_2|+|c_3|\leq1$, is the separable subset~\cite{peres1996separability}. Thus, if $p_{\max}=\max_k p_k$, entanglement is present precisely when $p_{\max}>1/2$, and the concurrence is $\max\{0,2p_{\max}-1\}$~\cite{Wootters1998}. By contrast, CHSH nonlocality requires the stricter condition $\max_{\alpha<\beta}(c_\alpha^2+c_\beta^2)>1$~\cite{Horodecki1995}. Hence, an entangled Bell-diagonal state need not violate the CHSH inequality. This family therefore provides a standard solvable setting for entanglement~\cite{Horodecki2009}, CHSH nonlocality, and quantum discord~\cite{Luo2008,LangCaves2010}. Local Pauli noise preserves the Bell-diagonal form and reduces the open-system evolution to the decay of a small number of correlation coefficients. This property makes Bell-diagonal states particularly convenient for analyzing the properties of KD trajectory quasiprobabilities.

\subsection{\label{dephasing channel}Quasiprobability trajectories under local dephasing}

The single-qubit dephasing channel characterizes the decay of off-diagonal coherence in the basis $\{|0\rangle,|1\rangle\}$ \cite{NielsenChuang2010}. It has the Kraus operators
\begin{equation}
\label{eq:karus_dephasing}
K_0=\sqrt{\lambda}\,\mathbb{I},\quad K_1=\sqrt{1-\lambda}\ket{0}\!\bra{0},\quad K_2=\sqrt{1-\lambda}|1\rangle\langle 1|,
\end{equation}
where $0\leq \lambda \leq 1$ is the dephasing parameter. The channel leaves the diagonal matrix elements unchanged and multiplies each off-diagonal element by $\lambda$. Thus, $\lambda=1$ gives the identity channel, whereas $\lambda=0$ removes all coherence in the computational basis.

Applying the same dephasing channel independently to both qubits yields the final state
\begin{multline}
\rho'_S=\frac{1}{4}\Bigl(
\mathbb{I}_{S_1}\otimes\mathbb{I}_{S_2}
+\lambda^2c_1\,\sigma^1_{S_1}\otimes\sigma^1_{S_2}\\
{}+\lambda^2c_2\,\sigma^2_{S_1}\otimes\sigma^2_{S_2}
+c_3\,\sigma^3_{S_1}\otimes\sigma^3_{S_2}
\Bigr).
\end{multline}
The state therefore remains Bell-diagonal. Specifically, the transverse correlations $c_1$ and $c_2$ acquire the decay factor $\lambda^2$, while the longitudinal correlation $c_3$ is unchanged. The local marginals remain $\rho'_{S_1}=\rho'_{S_2}=\mathbb{I}/2$. The final global eigenvectors are the Bell states in Eq.~\eqref{eq:Bell-states}, with populations
\begin{subequations}
\label{eq:pk-prime}
\begin{align}
p_{k'=\Phi^\pm}
&=\frac{1}{4}\left(1+c_3\pm\lambda^2v\right),\\
p_{k'=\Psi^\pm}
&=\frac{1}{4}\left(1-c_3\pm\lambda^2u\right).
\end{align}
\end{subequations}
Recall that $u=c_1+c_2$ and $v = c_1-c_2$. 

For the specified computational-basis local projector, direct calculations show that the transition KD quasiprobability $q[\gamma]$ in Eq.~\eqref{eq:q-operator-sum} is nonzero only if $s'=s$. For fixed $k$ and $s$, the only possible final Bell states lie in the same parity sector as $k$. Correspondingly, there are only two values of the transition KD quasiprobability, given by
\begin{equation}
q[\gamma]=
\begin{cases}
\dfrac{1+\lambda^2}{4},&k'=k,\\[5pt]
\dfrac{1-\lambda^2}{4},&k'\neq k.
\end{cases}
\label{eq:dephasing-q-values}
\end{equation}
The larger value of $q[\gamma]$ corresponds to trajectories that remain in the initial Bell state, whereas the smaller value corresponds to trajectories that switch to its partner within the same Bell pair. Correspondingly, the KD quasiprobabilities $\mathcal Q_S[\gamma]$ given by
\begin{equation}
\mathcal Q_S[\gamma]
=\frac{p_k}{4}\left(1\pm\lambda^2\right)\geq0,
\label{eq:QS-dephasing}
\end{equation}
are nonnegative. It is straightforward to verify that summing the KD quasiprobability weights $\mathcal Q_S[\gamma]$ corresponding to the stochastic mutual information change $\Delta\iota[\gamma]$ recovers both the average mutual information dissipation in Eq.~\eqref{eq:Qs_avg} and the integral FT in Eq.~\eqref{eq:Qs_fluct}. Figure~\ref{fig:result1}(a) illustrates the nonnegative trajectory weights for a representative Bell-diagonal state.

\begin{figure}
\centering
\includegraphics[width=\linewidth]{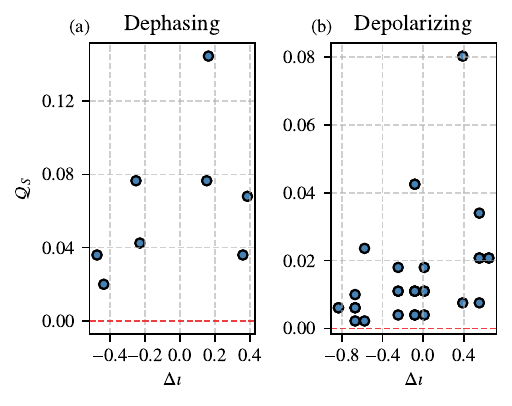}
\caption{The quasiprobability distribution $\mathcal Q_S[\gamma]$ of stochastic mutual information change $\Delta\iota[\gamma]$ for initial Bell-diagonal states with correlation coefficients $c_1=-0.3$, $c_2=-0.1$, and $c_3=-0.3$. The dynamics of each qubit is described by (a) dephasing with $\lambda=0.6$ and (b) depolarization with $\eta=7/15$. The red dashed line marks zero.}
\label{fig:result1}
\end{figure}

\subsection{\label{depolasirprobability}Quasiprobability trajectories under local depolarizing noise}

Next we consider the depolarizing channel, a typical single-qubit channel that drives a state towards the completely mixed state \cite{NielsenChuang2010}. Its Kraus operators are
\begin{equation}
K_0=\sqrt{\frac{1+3\eta}{4}}\,\mathbb{I},
\qquad
K_\alpha=\sqrt{\frac{1-\eta}{4}}\,\sigma^\alpha,
\end{equation}
where $\alpha=1,2,3$, while complete positivity requires $-1/3\leq\eta\leq1$. The depolarizing channel rescales the Pauli matrices, namely $\mathcal E(\sigma^\alpha)=\eta\,\sigma^\alpha$. It therefore preserves the Bell-diagonal structure and gives the final state
\begin{multline}
\rho'_S = \frac{1}{4}\Bigl(
\mathbb{I}_{S_1}\otimes\mathbb{I}_{S_2}
+\eta^2c_1\,\sigma^1_{S_1}\otimes\sigma^1_{S_2}\\
{}+\eta^2c_2\,\sigma^2_{S_1}\otimes\sigma^2_{S_2}
+\eta^2c_3\,\sigma^3_{S_1}\otimes\sigma^3_{S_2}
\Bigr),
\end{multline}
which yields a homogeneous contraction of the correlation coefficients $c_\alpha$. Correspondingly, the final Bell-state populations are
\begin{subequations}
\label{eq:p_eigenvalue_depolarzing}
\begin{align}
p_{k'=\Phi^\pm}
&=\frac{1}{4}\left(1+\eta^2(c_3\pm v)\right),\\
p_{k'=\Psi^\pm}
&=\frac{1}{4}\left(1+\eta^2(-c_3\pm u)\right).
\end{align}
\end{subequations}
The common factor $\eta^2$ multiplying the parameters $c_3$, $u$, and $v$ reflects the channel's tendency to drive the state towards the maximally mixed state, erasing all information about the initial state.

Direct calculation shows that the transition KD quasiprobability $q[\gamma]$ can take only four distinct nonzero values,
\begin{subequations}
\label{eq:kappa_parameter}
\begin{align}
\kappa_1&=\frac{1+2\eta+5\eta^2}{16},\\
\kappa_2&=\frac{(1-\eta)(1+3\eta)}{16},\\
\kappa_3&=\frac{(1-\eta)^2}{16},\\
\kappa_4&=\frac{1-\eta^2}{16}.
\end{align}
\end{subequations}
All four coefficients are nonnegative throughout the admissible range $-1/3\leq\eta\leq1$.

The four distinct values of $q[\gamma]$ correspond to distinct changes in the global Bell-state label and the local computational-basis label. The coefficient $\kappa_1$ describes trajectories that preserve both labels, for example, $q[\Phi^+,00,\Phi^+,00]=\kappa_1$. The coefficient $\kappa_2$ corresponds to a sign change within a Bell pair, $\Phi^+\leftrightarrow\Phi^-$ or $\Psi^+\leftrightarrow\Psi^-$, while the local label is unchanged. A representative trajectory is $q[\Phi^+,00,\Phi^-,00]=\kappa_2$. For $\kappa_3$, the Bell parity sector is unchanged, but the sign label may change. The local label is exchanged between the two correlated computational-basis states in the same Bell sector, with $00\leftrightarrow11$ in the $\Phi$ sector and $01\leftrightarrow10$ in the $\Psi$ sector. For instance, $q[\Phi^+,00,\Phi^+,11]=\kappa_3$. Finally, $\kappa_4$ describes transitions between the $\Phi$ and $\Psi$ sectors, accompanied by an exchange between the local-label sets $\{00,11\}$ and $\{01,10\}$. One example is $q[\Phi^+,00,\Psi^+,01]=\kappa_4$. The complete list of nonzero trajectories is given in Appendix~\ref{app:depolarizing_quasiprobability}.

Combining the transition KD quasiprobability $q[\gamma]$ with the initial Bell-state populations, we obtain the complete KD quasiprobability trajectories
\begin{equation}
\mathcal Q_S[\gamma]=p_k q[\gamma]\geq0,
\qquad
q[\gamma]\in\{\kappa_1,\kappa_2,\kappa_3,\kappa_4\},
\label{eq:QS-depolarizing}
\end{equation}
which are all nonnegative. Figure~\ref{fig:result1}(b) illustrates the nonnegative trajectory weights of $\mathcal Q_S[\gamma]$. Taken together, the dephasing and depolarizing examples show that the incompatibility of the global Bell projectors and the fixed local projectors does not by itself force the reduced quasiprobability to become negative, irrespective of the initial Bell-diagonal-state parameters.

Although individual trajectory weights depend on the local projector choice, the distribution of the information variable is unchanged for these two channels. Both initial and final local eigenvalues are $1/2$, so Eq.~\eqref{eq:stochastic-mutual information} reduces to $\Delta\iota[\gamma]=\ln(p_k/p_{k'})$, independently of $s,s'$. Completeness of the local projectors gives
\begin{equation}
\label{eq:local-label-coarse-graining}
\sum_{s,s'}\mathcal Q_S[k,s,k',s']=p_k\Tr\!\left[\Pi_{k'}\mathcal E_S(\Pi_k)\right]\geq0.
\end{equation}
Thus, after grouping trajectories with equal $\Delta\iota$, the distribution is real, nonnegative, and independent of the local rank-one projector. 

\section{\label{sec:decay}Quasiprobability distribution under the amplitude-damping channel}

In this section, we focus on the KD quasiprobability trajectories associated with the amplitude-damping channel in Sec.~\ref{amplitude damping channel}. We then analyze the trajectory weights by identifying the most negative transition KD quasiprobability $q[\gamma]$ in Sec.~\ref{Minimal negative value of the transition quasiprobability} and the most negative reduced KD quasiprobability trajectory $\mathcal Q_S[\gamma]$ in Sec.~\ref{minimal negative}.

\subsection{\label{amplitude damping channel}Quasiprobability trajectories under local amplitude damping}

The single-qubit amplitude-damping channel models energy relaxation from the excited state $\ket{1}$ to the ground state $\ket{0}$ through coupling to a zero-temperature environment, also known as a decay channel~\cite{NielsenChuang2010}. Its Kraus operators are
\begin{equation}
\label{eq:kraus_amplitude_damping}
K_0=\ket{0}\!\bra{0}+\sqrt{1-\mu}\,\ket{1}\!\bra{1},
\qquad
K_1=\sqrt{\mu}\,\ket{0}\!\bra{1},
\end{equation}
where $0\leq\mu\leq1$ is the damping probability and $K_0^\dagger K_0+K_1^\dagger K_1=\mathbb I$. We set $\bar{\mu}=1-\mu$. Applying the amplitude damping channel independently to the two qubits gives
\begin{equation}
\label{eq:rho'_S_decay}
\rho'_S=
\frac14
\begin{pmatrix}
\rho_{11} & 0 & 0 & \bar{\mu}v \\[3pt]
0 & \rho_{22} & \bar{\mu}u & 0 \\[3pt]
0 & \bar{\mu}u & \rho_{22} & 0 \\[3pt]
\bar{\mu}v & 0 & 0 & \rho_{44}
\end{pmatrix},
\end{equation}
where the scaled diagonal coefficients are
\begin{subequations}
\begin{align}
\rho_{11} &= c_3\bar{\mu}^2+(1+\mu)^2,\\
\rho_{22} &= \bar{\mu}\bigl(1-c_3\bar{\mu}+\mu\bigr),\\
\rho_{44} &= \bar{\mu}^2\bigl(1+c_3\bigr).
\end{align}
\end{subequations}
Recall that $u=c_1+c_2$ and $v=c_1-c_2$, as defined in Sec.~\ref{Bell diagonal state}. The final state is an $X$ state rather than, in general, a Bell-diagonal state. In particular, its local marginals are no longer maximally mixed, since the channel drives each qubit to the $|0\rangle$ state. Specifically, the two reduced states remain identical and diagonal in the computational basis, given by
\begin{equation}
\label{eq:decay_local_eigenvalues}
\rho'_{S_1}=\rho'_{S_2}=
\frac12
\begin{pmatrix}
1+\mu & 0 \\
0 & 1-\mu
\end{pmatrix}.
\end{equation}
We can directly read off the eigenvalues of $\rho'_{S_j}$ as $p_{s_j'=0}=(1+\mu)/2$ and $p_{s_j'=1}=(1-\mu)/2$ for $j=1,2$. For $\mu>0$, the nondegenerate local spectra uniquely select the computational-basis projectors $\Pi_{s_j'=0}=\ket{0}\!\bra{0}$ and $\Pi_{s_j'=1}=\ket{1}\!\bra{1}$. At $\mu=0$, the local states remain maximally mixed, and the same rank-one projector is retained by convention. 

Although the final state $\rho_S'$ in Eq. \eqref{eq:rho'_S_decay} breaks the Bell-diagonal structure, the $X$ structure admits an analytical treatment of the eigenvalues and eigenvectors. Specifically, the eigenvalues are
\begin{subequations}
\label{eq:decay_global_eigenvalues}
\begin{align}
p_{k'=\Psi^\pm}&=\frac{\bar{\mu}}{4}\left(2\pm u-\bar{\mu}\left(1+c_3\right)\right),\\
p_{k'=\Omega^\pm}&=\frac14\left(A\pm B\right),
\end{align}
\end{subequations}
with
\begin{equation}
\label{eq:A_and_B}
A=(1+c_3)\left(1+\mu^2\right)-2c_3\mu,
\quad
B=\sqrt{\bar{\mu}^2v^2+4\mu^2}.
\end{equation}
The eigenvalues labeled by $\Psi^\pm$ are associated with the Bell states $|\Psi^\pm\rangle$, whereas those labeled by $\Omega^\pm$ correspond to the eigenvectors
\begin{equation}
\label{eq:eigenvectors}
\ket{\Omega^+}=\omega^+_{00}\ket{00}+\omega^+_{11}\ket{11}, \quad \ket{\Omega^-}=\omega^-_{00}\ket{00}+\omega^-_{11}\ket{11},
\end{equation}
where
\begin{subequations}
\begin{align}
\omega^\pm_{00}&=\frac{2\mu\pm B}{\sqrt{(v\bar{\mu})^2+(2\mu\pm B)^2}},\\
\omega^\pm_{11}&=\frac{v\bar{\mu}}{\sqrt{(v\bar{\mu})^2+(2\mu\pm B)^2}}.
\end{align}
\end{subequations}
For $v=0$ and fixed $\mu>0$, the expressions for $\ket{\Omega^\pm}$ are understood in the continuous limit, in which they become the computational states.

With the initial computational-basis local projector fixed as in Sec.~\ref{Bell diagonal state}, the spectral information above determines the reduced KD quasiprobability trajectories in Eq.~\eqref{eq:Q_S}. The full analytical list of the nonzero KD quasiprobability trajectories is presented in Appendix~\ref{app:decay_quasiprobability}. Most of the weights are nonnegative regardless of the parameter values. For $0<\mu<1$, we find that two trajectories can be negative whenever $v=c_1-c_2\neq 0$. Specifically, for $v>0$, the two negative trajectories are
\begin{subequations}
\label{eq:Q_S_v>0}
\begin{align}
\label{eq:Q_S_v>0_1}
\mathcal Q_S[\Phi^+,00,\Omega^-,00]
&=\frac12\left(\bigl(\omega^-_{00}\bigr)^2+\bar{\mu}\omega^-_{00}\omega^-_{11}\right)p_{k=\Phi^+},\\
\label{eq:Q_S_v>0_2}
\mathcal Q_S[\Phi^-,11,\Omega^+,11]
&=\frac12\left(-\bar{\mu}\omega^+_{00}\omega^+_{11}+\bar{\mu}^2\bigl(\omega^+_{11}\bigr)^2\right)p_{k=\Phi^-}.
\end{align}
\end{subequations}
For $v<0$, the negative trajectories are
\begin{subequations}
\label{eq:Q_S_v<0}
\begin{align}
\mathcal Q_S[\Phi^-,00,\Omega^-,00]
&=\frac12\left(\bigl(\omega^-_{00}\bigr)^2-\bar{\mu}\omega^-_{00}\omega^-_{11}\right)p_{k=\Phi^-},\\
\mathcal Q_S[\Phi^+,11,\Omega^+,11]
&=\frac12\left(\bar{\mu}\omega^+_{00}\omega^+_{11}+\bar{\mu}^2\bigl(\omega^+_{11}\bigr)^2\right)p_{k=\Phi^+}.
\end{align}
\end{subequations}
See Fig.~\ref{fig:result2} for the negative KD quasiprobability trajectory weights. 
Examples with other local bases, illustrating the basis dependence of the trajectory weights, are given in Appendix~\ref{app:local-resolution}.

\begin{figure}
    \centering
    \includegraphics[width=1\linewidth]{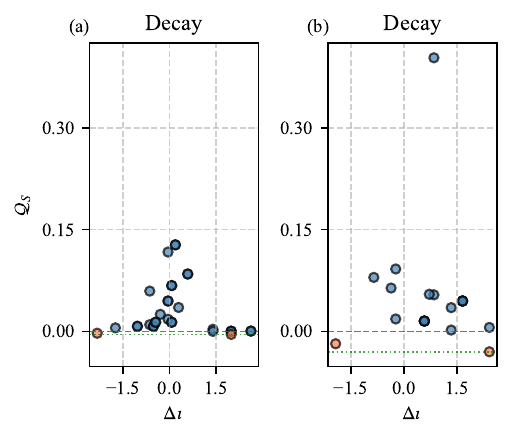}
    \caption{The reduced KD quasiprobability distribution $\mathcal Q_S[\gamma]$ of the stochastic mutual information change $\Delta\iota[\gamma]$ for Bell-diagonal states under local amplitude damping. In both panels, the damping parameter is $\mu=0.4$. Orange points identify negative trajectory weights, and the green dotted line in each panel marks the corresponding minimum trajectory weight. Panel (a) uses the representative state $c_1=c_3=-0.3$ and $c_2=-0.1$. Panel (b) uses the $v>0$ minimizing family in Theorem~\ref{thm:global-most-negative}, with $c_3=1$ and $c_1=-c_2\simeq0.49201$. }
    \label{fig:result2}
\end{figure}

For $v=0$, the initial even-parity coherence $\bra{00}\rho_S\ket{11}=v/4$ vanishes and remains absent under the amplitude-damping channel. The final even-parity block is then diagonal in the computational basis, so $\ket{\Omega^+}$ and $\ket{\Omega^-}$ reduce to $\ket{00}$ and $\ket{11}$, respectively, and $\omega^\pm_{00}\omega^\pm_{11}=0$. The final even-parity eigenprojectors therefore commute with the local computational-basis projectors, and all trajectory weights are nonnegative.

For example, the maximally entangled and Bell-nonlocal state $\ket{\Psi^+}$ has $(c_1,c_2,c_3)=(1,1,-1)$ and hence $v=0$, so it does not generate the negative trajectories identified here. Conversely, the separable Bell-diagonal state $(c_1,c_2,c_3)=(0.3,0,0)$ lies inside the separable octahedron but has $v=0.3\neq0$ and does generate such trajectories for $0<\mu<1$. This explicit contrast shows that the trajectory negativity is not determined by either entanglement or CHSH nonlocality.

\subsection{\label{Minimal negative value of the transition quasiprobability}Most negative transition quasiprobability trajectory}

The reduced KD quasiprobability trajectory weight can be written as $\mathcal Q_S[\gamma]=p_kq[\gamma]$. Its negativity arises solely from the transition KD quasiprobability $q[\gamma]$. It is therefore natural to ask how negative $q[\gamma]$ can become. The minimization of the negative trajectory weights is symmetric under the exchange $c_1\leftrightarrow c_2$, which maps $v$ to $-v$ while leaving $u$ and $c_3$ invariant and exchanging $\Phi^+\leftrightarrow\Phi^-$. The $v<0$ trajectories are therefore in one-to-one correspondence with those for $v>0$. In the following, we restrict the analysis to the regime $v>0$. 

Formally, we find that the most negative transition KD quasiprobability $q[\gamma]$ is given by the following parameters. 

\begin{theorem}
\label{thm:global-transition-minimum}
The smallest negative transition quasiprobability $q[\gamma]$ is given by the trajectory $q[\Phi^-,11,\Omega^+,11]$ at $v = v_{q_{\min}}$ with
\begin{equation}
\label{eq:transition-q-minimizer}
v_{q_{\min}}=
\begin{cases}
\dfrac{2\mu}{\bar{\mu}^2}, & 0<\mu\leq\mu_q,\\[6pt]
2, & \mu_q\leq\mu<1,
\end{cases}
\end{equation}
where $\mu_q=(3-\sqrt 5)/2$. The corresponding minimum value $q_{\min} = q[\Phi^-,11,\Omega^+,11]$ is
\begin{equation}
\label{eq:transition-q-min-value}
q_{\min}=
\begin{cases}
\dfrac{\bar{\mu}}{4}
\left(\bar{\mu}-\sqrt{1+\bar{\mu}^2}\right),
& 0<\mu\leq\mu_q,\\[8pt]
\dfrac{\bar{\mu}^2
\left(\sqrt{\mu^2+\bar{\mu}^2}-1-\mu\right)}
{4\sqrt{\mu^2+\bar{\mu}^2}},
& \mu_q\leq\mu<1.
\end{cases}
\end{equation}
\end{theorem}

The proof proceeds by comparing the two negative trajectories to identify $q[\Phi^-,11,\Omega^+,11]$ as the lower one and then minimizing it over the parameter regime $0<v\leq2$. A complete proof is given in Appendix~\ref{app:proof_theorem_1}.

For $0<\mu<1$, Eq.~\eqref{eq:transition-q-min-value} shows that $q_{\min}$ increases monotonically from the limiting value $(1-\sqrt{2})/4$ as $\mu\to0^+$ towards zero as $\mu\to1^-$. Thus, the largest magnitude of the negative transition coefficient is approached, but not attained, for arbitrarily weak nonzero damping. At $\mu=0$, the identity channel gives nonnegative weights when the initial and final global Bell projectors and local computational-basis projectors are chosen identically. The transition coefficient depends on $v$ and $\mu$, whereas the complete weight also includes the initial-state population. In the representative family $c_3=1$, $u=0$, the optimum at $v=2$ for $\mu\geq\mu_q$ has $p_{\Phi^-}=0$, so its complete weight vanishes. Along the same family as $\mu\to0^+$, $v_{q_{\min}}\to0$ and $p_{\Phi^-}\to1/2$, giving the complete-weight limit $(1-\sqrt2)/8$. Note that the transition KD quasiprobability has the parameter-independent lower bound $q[\gamma]>(1-\sqrt{2})/4\approx -0.10355$.

For any fixed $v>0$, the negative coefficients satisfy $q_\pm=O(\mu)$, so their complete weights vanish as $\mu\to0^+$. Along the optimizing families, however, $v=2\mu+O(\mu^2)$. Writing $\ket{\Omega^+}=\cos\theta\ket{00}+\sin\theta\ket{11}$, with $0<\theta<\pi/4$, gives $\tan(2\theta)=\bar\mu v/(2\mu)$: fixed $v>0$ yields $\theta\to\pi/4$, whereas optimization yields $\theta\to\pi/8$. This path dependence of the rank-one projectors at degeneracy accounts for the finite limit. The final even-parity eigenvalues $p_{k'=\Omega^\pm}$ differ by $B/2=O(\mu)$ along optimization, where $B=\sqrt{\bar\mu^2v^2+4\mu^2}$ given by Eq.~\eqref{eq:A_and_B}. Keeping $u$ and $c_3$ fixed, perturbations satisfying $\delta v/\mu\to0$ preserve the weak-damping limits of the $\Omega^\pm$ projectors and negative weights. Perturbations of order $\mu$ can change these limits. For example, adding $\delta v=\mu$ changes the limiting ratio $v/\mu$ from $2$ to $3$. A fixed admissible offset $\delta v=\epsilon>0$, independent of $\mu$, instead gives $v\to\epsilon$ as $\mu\to0^+$, so the negative weights tend to zero as in the fixed-$v$ case.

\subsection{\label{minimal negative}Most negative reduced quasiprobability trajectory}

The trajectory $[\Phi^-,11,\Omega^+,11]$ with the most negative transition KD quasiprobability $q[\Phi^-,11,\Omega^+,11]$ does not yield the most negative reduced quasiprobability $\mathcal Q_S[\gamma]$. For $v>0$, Eq.~\eqref{eq:pk} gives $p_{\Phi^-}<p_{\Phi^+}$, so the initial-state populations shift the minimum of $\mathcal Q_S[\gamma]=p_kq[\gamma]$ to a different trajectory. We therefore optimize the full reduced KD quasiprobability trajectory separately. The $v<0$ case follows by the exchange $c_1\leftrightarrow c_2$, which maps $v$ to $-v$ and relabels $\Phi^+$ and $\Phi^-$. Specifically, the most negative KD quasiprobability trajectory $\mathcal Q_S[\gamma]$ is given by the following theorem. 
\begin{theorem}
\label{thm:global-most-negative}
For $v>0$, the most negative reduced KD quasiprobability trajectory is $\mathcal Q_S[\Phi^+,00,\Omega^-,00]$, attained for $c_1=-c_2=v_{\mathcal Q_{\min}}/2$ and $c_3=1$, with
\begin{equation}
\label{eq:rstar-definition}
v_{\mathcal Q_{\min}}=
\begin{cases}
\nu_0, & 0<\mu<\mu_c,\\[4pt]
2, & \mu_c\leq\mu<1,
\end{cases}
\end{equation}
where $\mu_c\approx 0.82668$ and $\nu_0$ is the unique solution in $(2\mu,2)$ of
\begin{equation}
\label{eq:nu_0_equation}
1-\frac{2\mu+\bar{\mu}^2v}{B(v)}
-\frac{2\bar{\mu}^2\mu(2+v)(2\mu-v)}{B^3(v)}
=0.
\end{equation}
The $v$-dependent parameter $B(v)$ is defined in Eq. \eqref{eq:A_and_B}. 
\end{theorem}

The proof is given in Appendix~\ref{app:negative}. Once the minimizing parameter $v_{\mathcal Q_{\min}}$ has been determined, the corresponding minimum weight follows by direct substitution, which gives
\begin{equation}
\mathcal Q_{S,\min}
=
\begin{cases}
\dfrac{(2+\nu_0)\left(B(\nu_0)-2\mu-\bar{\mu}^2\nu_0\right)}
{16B(\nu_0)},
& 0<\mu<\mu_c,\\[8pt]
\dfrac{\sqrt{\bar{\mu}^2+\mu^2}-\mu-\bar{\mu}^2}
{4\sqrt{\bar{\mu}^2+\mu^2}},
& \mu_c\leq\mu<1,
\end{cases}
\label{eq:global-Qmin}
\end{equation}
where $\nu_0$ is the unique solution of Eq. \eqref{eq:nu_0_equation}. The optimal parameter satisfies $0<v<2$ for $0<\mu<\mu_c$ and reaches $v=2$ at $\mu=\mu_c$, while the minimizing state always lies on a Bell-tetrahedron edge. For $\mu_c\leq\mu<1$, the minimizing state is the pure Bell state $\ket{\Phi^+}\!\bra{\Phi^+}$.

Physically, the global minimum results from the competition between the negative transition KD quasiprobability $q[\gamma]$ and its associated initial Bell-state population $p_k$. Although $q[\Phi^-,11,\Omega^+,11]$ can be more negative, its smaller factor $p_{\Phi^-}$ prevents it from determining the minimum weight of $\mathcal Q_S[\gamma]$. The trajectory $[\Phi^+,00,\Omega^-,00]$ instead combines a negative contribution with the larger population $p_{\Phi^+}$. Figure~\ref{fig:result2}(a) illustrates a representative state whose negative points lie above the global minimum. For $\mu=0.4$, the theorem gives $v_{\mathcal Q_{\min}}\approx 0.98402$ and hence $c_1=-c_2\approx 0.49201$. This representative case is illustrated in Fig.~\ref{fig:result2}(b), where the lowest orange point is the trajectory $\mathcal Q_S[\Phi^+,00,\Omega^-,00]$, which is appreciably lower than the lowest point in Fig.~\ref{fig:result2}(a). 

Numerical evaluation of Eq.~\eqref{eq:global-Qmin} shows that $\mathcal Q_{S,\min}$ increases monotonically with $\mu$ and approaches $(1-\sqrt{2})/8\approx-0.05178$ as $\mu\to0^+$. For every Bell-diagonal initial state and $0<\mu<1$, the trajectory weights therefore obey the parameter-independent bound $\mathcal Q_S[\gamma]>(1-\sqrt{2})/8$. Along the minimizing family, $\nu_0=2\mu+O(\mu^2)$ and $p_{\Phi^+}\to1/2$. The transition-coefficient optimum has the same leading scaling of $v$ but selects the other negative trajectory. For the representative family $c_3=1$, $u=0$, its population $p_{\Phi^-}$ also tends to $1/2$. Both paths therefore approach the same complete-weight limit, although their minimizers differ at finite $\mu$.

\section{\label{sec:neagative}Conditions for negative quasiprobability trajectories}

In Sec.~\ref{Interference decomposition of the quasiprobability}, we derive a diagonal-interference decomposition of $q[\gamma]$ and bounds connecting the channel matrix elements to the final global eigenvector. We then present another example in Sec.~\ref{Examples on negative quasiprobability} that explicitly yields a strictly negative reduced KD quasiprobability trajectory $\mathcal Q_S[\gamma]$.

\subsection{Interference decomposition of the transition quasiprobability trajectory \label{Interference decomposition of the quasiprobability}}

The channel-dependent results in Secs. \ref{sec:dephase_depolarizing} and \ref{sec:decay} show that incompatibility between the global and local projectors does not imply negative reduced KD quasiprobabilities $\mathcal Q_S[\gamma]$. This motivates us to identify the mechanism underlying negative weights.  

For a trajectory $\gamma=[k,s,k',s']$, Eq.~\eqref{eq:QS-q-pk} gives $\mathcal Q_S[\gamma]=p_kq[\gamma]$. Since $p_k\geq0$, a negative transition quasiprobability $q[\gamma]$ produces a negative reduced weight whenever $p_k>0$. Conversely, a trajectory with $p_k=0$ carries zero reduced weight independently of the value assigned to its transition quasiprobability. We therefore focus on analyzing the sign of $q[\gamma]$. Recall that $\{\ket{s'}\}$ is the final local product eigenbasis. The final global eigenvector $\ket{k'}$ can be expanded as
\begin{equation}
\ket{k'}=z_{s'}\ket{s'}+\sum_{r'\neq s'}z_{r'}\ket{r'},
\label{eq:final-global-expansion}
\end{equation}
where $z_{r'}=\langle r'|k'\rangle$. Substitution into the amplitude representation in Eq.~\eqref{eq:q-operator-sum}, together with $\langle k'|s'\rangle=z_{s'}^*$, yields the decomposition
\begin{equation}
\label{eq:interference-decomposition}
q[\gamma] = q_\text{D}[\gamma] + q_\text{I}[\gamma],
\end{equation}
with
\begin{subequations}
\label{eq:interference-q}
\begin{align}
\label{eq:interference-D}
q_\text{D}[\gamma]
&=|z_{s'}|^2\bra{s'}\mathcal E_S(\Pi_s\Pi_k)\ket{s'}, \\
\label{eq:interference-I}
q_\text{I}[\gamma]
&=z_{s'}^*\sum_{r'\neq s'}
\bra{s'}\mathcal E_S(\Pi_s\Pi_k)\ket{r'}z_{r'}.
\end{align}
\end{subequations}
The term $q_\text{D}[\gamma]$ is diagonal in the final local basis, while $q_\text{I}[\gamma]$ collects the off-diagonal contributions that depend on the coherence of the final global eigenvector in the local basis. Neither contribution is generally a classical probability because $\Pi_s\Pi_k$ need not be Hermitian or positive.

Motivated by the channel-dependent results in Secs. \ref{sec:dephase_depolarizing} and \ref{sec:decay}, we formulate the following sufficient condition for negative quasiprobability trajectories. 
\begin{proposition}
\label{prop:interference-dominance}
Consider a trajectory $\gamma$ with $p_k>0$ for which both $q_\text{D}[\gamma]$ and $q_\text{I}[\gamma]$ are real. Given the decomposition $q[\gamma] = q_\text{D}[\gamma] + q_\text{I}[\gamma]$ in Eq.~\eqref{eq:interference-decomposition}, if $q_\text{D}[\gamma]\geq0$ and $q_\text{I}[\gamma]<-q_\text{D}[\gamma]$, then $q[\gamma]<0$ and $\mathcal Q_S[\gamma]=p_kq[\gamma]<0$.
\end{proposition}

We refer to the condition in Proposition~\ref{prop:interference-dominance} as interference dominance, which states that a negative off-diagonal contribution must outweigh the nonnegative diagonal contribution for a negative trajectory weight to occur. The resulting negative distribution is therefore not a consequence of noncommutativity alone. Instead, it results from dynamical destructive interference between the local-basis components of the final global eigenstate.

Let $\Delta$ denote complete dephasing in the specified initial local product basis, defined for any operator $X$ by
\begin{equation}
\Delta(X)=\sum_s\Pi_s X\Pi_s.
\end{equation}
The map $\Delta'$ denotes complete dephasing in the corresponding final local product basis.
Suppose the CPTP channel satisfies the detection-incoherence condition~\cite{Liu2017ResourceDestroyingMaps,Theurer2019DetectionCoherence}
\begin{equation}
\label{eq:Delta_condition}
\Delta'\mathcal E_S=\Delta'\mathcal E_S\Delta,
\end{equation}
so input coherences do not affect output populations. With the conditional transition probability $p_{s\rightarrow s'}=\langle s'|\mathcal E_S(\Pi_s)|s'\rangle$ and $\Delta(\Pi_s\Pi_k)=|\langle s|k\rangle|^2\Pi_s$, we obtain
\begin{equation}
\label{eq:population-reference}
q_{\rm D}[\gamma]=|z_{s'}|^2|\langle s|k\rangle|^2p_{s\rightarrow s'}\geq0.
\end{equation}
Then applying the Cauchy--Schwarz inequality to Eq.~\eqref{eq:interference-I}, with $\sum_{r'\ne s'}|z_{r'}|^2=1-|z_{s'}|^2$, we obtain the bound
\begin{equation}
\label{eq:channel-interference-bound}
|q_{\rm I}[\gamma]|\leq |z_{s'}|\sqrt{1-|z_{s'}|^2}\left(\sum_{r'\ne s'}\left|\langle s'|\mathcal E_S(\Pi_s\Pi_k)|r'\rangle\right|^2\right)^{1/2}.
\end{equation}
All channel examples considered in the main text satisfy the condition in Eq. \eqref{eq:Delta_condition} in the computational basis. In this case, $q_{\rm D}$ is a normalized conditional reference weight for sequential measurements $s\to s'\to k'$ at fixed $k$, with the channel acting between $s$ and $s'$. Since $\operatorname{Re}q[\gamma]\geq q_{\rm D}[\gamma]-|q_{\rm I}[\gamma]|$, Eq.~\eqref{eq:channel-interference-bound} excludes negative real weights whenever its right-hand side does not exceed $q_{\rm D}$. Negativity with $p_k>0$ requires this upper bound to exceed $q_{\rm D}$ and the relative phases to be destructive.

For the chosen computational-basis local projector, the dephasing and depolarizing channels provide two useful counterexamples to the inference that noncommutativity must produce negativity. For local dephasing, the two transition coefficients in Eq.~\eqref{eq:dephasing-q-values} have the decompositions
\begin{subequations}
\label{eq:DI-dephasing}
\begin{align}
\left(q^{(1)}_\text{D}[\gamma],q^{(1)}_\text{I}[\gamma]\right) & = \left(\frac14,\frac{\lambda^2}{4}\right), \\
\left(q^{(2)}_\text{D}[\gamma],q^{(2)}_\text{I}[\gamma]\right) & = \left(\frac14,-\frac{\lambda^2}{4}\right).
\end{align}
\end{subequations}
Here, the superscript distinguishes the distinct values of $q_\text{D}[\gamma]$ and $q_\text{I}[\gamma]$. Although $q_\text{I}[\gamma]$ can be negative, its magnitude never exceeds $q_\text{D}[\gamma]$ for $0\leq\lambda\leq1$. Hence, the destructive contribution can reduce a trajectory weight to zero only at the boundary points and cannot make it negative.

For local depolarization, the four transition coefficients in Eq.~\eqref{eq:kappa_parameter} have the decompositions
\begin{subequations}
\label{eq:DI-depolarizing}
\begin{align}
\left(q^{(1)}_\text{D}[\gamma],q^{(1)}_\text{I}[\gamma]\right)
&=\left(\frac{(1+\eta)^2}{16},\frac{\eta^2}{4}\right),\\
\left(q^{(2)}_\text{D}[\gamma],q^{(2)}_\text{I}[\gamma]\right)
&=\left(\frac{(1+\eta)^2}{16},-\frac{\eta^2}{4}\right),\\
\left(q^{(3)}_\text{D}[\gamma],q^{(3)}_\text{I}[\gamma]\right)
&=\left(\frac{(1-\eta)^2}{16},0\right),\\
\left(q^{(4)}_\text{D}[\gamma],q^{(4)}_\text{I}[\gamma]\right)
&=\left(\frac{1-\eta^2}{16},0\right).
\end{align}
\end{subequations}
The only negative interference contribution is $q^{(2)}_\text{I}[\gamma]$. The strict inequality $q^{(2)}_\text{I}[\gamma]<-q^{(2)}_\text{D}[\gamma]$ would require $(\eta-1)(3\eta+1)>0$, which is incompatible with the complete-positivity interval $-1/3\leq\eta\leq1$. Equality occurs only at the two boundary points, where $q^{(2)}_\text{D}[\gamma]+q^{(2)}_\text{I}[\gamma]=0$. Thus, depolarization can suppress this transition weight but cannot produce a negative one.

For amplitude damping, the relative sizes of channel matrix elements and final global eigenvector components allow destructive interference to dominate. For $0<\mu<1$ and $v>0$, the negative trajectories $[\Phi^+,00,\Omega^-,00]$ and $[\Phi^-,11,\Omega^+,11]$ given by Eq.~\eqref{eq:Q_S_v>0} have the following decompositions, respectively,
\begin{subequations}
\label{eq:DI-AD-positive}
\begin{align}
\left(q^{(1)}_\text{D}[\gamma],q^{(1)}_\text{I}[\gamma]\right)
&=\left(\frac12\bigl(\omega^-_{00}\bigr)^2,\frac12\bar{\mu}\omega^-_{00}\omega^-_{11}\right),\\
\left(q^{(2)}_\text{D}[\gamma],q^{(2)}_\text{I}[\gamma]\right)
&=\left(\frac12\bar{\mu}^2\bigl(\omega^+_{11}\bigr)^2,-\frac12\bar{\mu}\omega^+_{00}\omega^+_{11}\right).
\end{align}
\end{subequations}
For $v<0$, the exchange $\Phi^+\leftrightarrow\Phi^-$ gives the corresponding trajectories $[\Phi^-,00,\Omega^-,00]$ and $[\Phi^+,11,\Omega^+,11]$, with
\begin{subequations}
\label{eq:DI-AD-negative}
\begin{align}
\left(q^{(3)}_\text{D}[\gamma],q^{(3)}_\text{I}[\gamma]\right)
&=\left(\frac12\bigl(\omega^-_{00}\bigr)^2,-\frac12\bar{\mu}\omega^-_{00}\omega^-_{11}\right),\\
\left(q^{(4)}_\text{D}[\gamma],q^{(4)}_\text{I}[\gamma]\right)
&=\left(\frac12\bar{\mu}^2\bigl(\omega^+_{11}\bigr)^2,\frac12\bar{\mu}\omega^+_{00}\omega^+_{11}\right).
\end{align}
\end{subequations}
The coefficients in Eq.~\eqref{eq:eigenvectors} give $q^{(j)}_{\rm I}<-q^{(j)}_{\rm D}$ for the respective $v$ regimes. Only one relevant off-diagonal term contributes for these two-component eigenvectors, so the magnitude bound in Eq.~\eqref{eq:channel-interference-bound} is saturated. The destructive phase yields the negative trajectories.

\subsection{Example of negative quasiprobability trajectories}\label{Examples on negative quasiprobability}

\noindent The following example illustrates the interference bound and its phase-sensitive threshold. Consider the single-qubit channel $\mathcal E$ with Kraus operators
\begin{equation}
\label{eq:example-kraus}
K_0=\frac{1}{\sqrt{2}}\ket0\!\bra0,\qquad K_1=\ket0\!\bra1+\frac{1}{\sqrt{2}}\ket1\!\bra0,
\end{equation}
which satisfy $K_0^\dagger K_0+K_1^\dagger K_1=\mathbb I$. The channel is a pre-rotated amplitude-damping channel, obtained by applying a bit flip before standard amplitude damping with damping parameter \(1/2\). Physically, the map describes a \(\pi\) pulse that interchanges the computational-basis populations immediately before energy relaxation~\cite{Fisher2012}.

We take the initial state to be $\rho_S=\ket{\Phi^+}\!\bra{\Phi^+}$, whose only nonzero Bell-state eigenvalue is $p_{k=\Phi^+}=1$, and select the initial labels $k=\Phi^+$ and $s=11$. Applying the product channel gives
\begin{equation}
(\mathcal E\otimes\mathcal E)(\Pi_{s=11}\Pi_{k=\Phi^+})
=\frac12\ket{00}\!\bra{00}
+\frac14\ket{00}\!\bra{11}.
\label{eq:example-evolved-operator}
\end{equation}
The first term is diagonal in the computational basis and the second one is off-diagonal, thereby clearly identifying the two contributions in Eq.~\eqref{eq:interference-decomposition}. The final state has two nondegenerate eigenvalues $\chi_\pm=(3\pm2\sqrt{2})/8$ with the corresponding orthonormal eigenvectors
\begin{align}
\ket{\chi_+}
&=\cos\frac{\pi}{8}\ket{00}
+\sin\frac{\pi}{8}\ket{11},\\
\ket{\chi_-}
&=-\sin\frac{\pi}{8}\ket{00}
+\cos\frac{\pi}{8}\ket{11}.
\end{align}
The two remaining eigenvalues are degenerate and equal to $1/8$, with eigenvectors $\ket{01}$ and $\ket{10}$. Tracing out either qubit gives $\rho'_{S_1}=\rho'_{S_2}=\text{diag}(3/4,1/4)$, so the final local eigenbasis remains the computational basis.

For the trajectory with final labels $k'=\chi_-$ and $s'=00$, we obtain the negative weight
\begin{equation}
q[\Phi^+,11,\chi_-,00] =\frac{4-3\sqrt2}{16}<0,
\end{equation}
which has the decomposition
\begin{equation}
q_\text{D} = \frac{2-\sqrt{2}}{8},\quad q_\text{I} = -\frac{\sqrt{2}}{16} < -q_\text{D}.
\end{equation}
The diagonal and off-diagonal matrix elements in Eq.~\eqref{eq:example-evolved-operator} are $1/2$ and $1/4$, respectively. With the destructive relative phase of $\ket{\chi_-}$, negativity occurs precisely when $|z_{11}/z_{00}|>2$, which is satisfied by $\cot(\pi/8)=1+\sqrt2$. The interference bound is saturated, and $p_{\Phi^+}=1$ gives $\mathcal Q_S=q<0$. 

\section{\label{sec:violation_horvath} Violation of Horv\'ath's criterion}

In this section, we examine whether the real quasiprobability weights satisfy Horv\'ath's necessary and sufficient conditions for the Jensen-Steffensen inequality. In Sec.~\ref{H-criterion}, we clarify the role of the standard Jensen's inequality in the quasiprobability FT and explain Horv\'ath's criterion. Section~\ref{Calculation of the Criterion} presents an example in which negative trajectory weights are neutralized by grouping at degenerate support values, so that Horv\'ath's criterion holds. Section~\ref{scend Calculation of the Criterion} then provides a contrasting case in which the parameter family containing the most negative KD quasiprobability trajectory violates Horv\'ath's criterion. Finally, Sec.~\ref{third Calculation of the Criterion} moves beyond this extremal family and identifies a structured amplitude-damping region in which the violation of Horv\'ath's criterion persists. 

\subsection{Jensen's inequality and Horv\'ath's criterion\label{H-criterion}}

For an ordinary probability distribution $\{p_l\}$, Jensen's inequality for the convex function $e^{-x}$ gives $\langle e^{-x}\rangle\geq e^{-\langle x\rangle}$ \cite{Jensen1906}. Here, Jensen's inequality denotes the standard nonnegative-weight result. The Jensen-Steffensen inequality is its classical extension to real weights under appropriate partial-sum constraints \cite{Steffensen1918}. Horv\'ath's criterion used below gives necessary and sufficient conditions for the corresponding discrete signed-weight problem. When the nonnegative weights $p_l$ are replaced by a normalized signed distribution $\{q_l\}$, the standard Jensen's inequality may, but need not, be violated, since positivity of $\{q_l\}$ is sufficient but not necessary for Jensen's inequality to hold. 

For strictly positive initial and final spectra, as specified in Sec.~\ref{subsec:kraus_quasiprobability}, the quasiprobability distribution, which may be signed or complex in general, satisfies the integral FT, but the standard Jensen's inequality cannot be applied directly. Therefore, the second-law-like inequality $\langle\Delta\iota\rangle_{\mathcal Q_S}\geq 0$ cannot be derived automatically. On the other hand, we have $\langle\Delta\iota\rangle_{\mathcal Q_S} = \Delta\mathcal I$ from Eq. \eqref{eq:classical-mutual information} and $\Delta\mathcal I\geq 0$ from the quantum data-processing inequality \cite{LiebRuskai1973,SchumacherNielsen1996,buscemi2014complete}. Therefore, the inequality
\begin{equation}
\label{eq:valid_jensen_inequality}
\langle e^{-\Delta\iota}\rangle_{\mathcal Q_S}=1\geq e^{-\langle\Delta\iota\rangle_{\mathcal Q_S}}
\end{equation}
must hold. This follows directly from the integral FT and the data-processing inequality, without applying the standard Jensen inequality to signed weights. The argument establishes Eq.~\eqref{eq:valid_jensen_inequality} for $f(x)=e^{-x}$; it does not ensure the Jensen--Steffensen inequality for every continuous convex function.

Recently, necessary and sufficient conditions for the discrete Jensen-Steffensen inequality with real weights were proposed by Horv\'ath \cite{Horvath2024}. We collectively refer to these conditions as Horv\'ath's criterion. Before introducing Horv\'ath's criterion, we first explain why a degenerate distribution must be treated separately. Suppose that the same value of a stochastic variable occurs with weights $q_1$ and $q_2$. When taking the average, this is equivalent to assigning the stochastic variable the weight $q_1+q_2$. Therefore, the negativity associated with a degenerate support value is invisible to the Jensen-Steffensen inequality. Formally, trajectories that share the same support value $\Delta\iota[\gamma]$ must first be grouped because only the weighted sums, rather than the individual trajectories, enter the inequality.

Degeneracy can occur in the stochastic mutual information change $\Delta\iota[\gamma]$ defined in Eq. \eqref{eq:stochastic-mutual information}. Let
\begin{equation}
\Delta\tilde{\iota}_1>\Delta\tilde{\iota}_2>\cdots>\Delta\tilde{\iota}_m
\label{eq:grouped-support-values}
\end{equation}
be the $m$ distinct values taken by $\Delta\iota[\gamma]$, ordered from largest to smallest, and define the corresponding grouped weights by
\begin{equation}
\tilde{\mathcal Q}_l=
\sum_{\tau:\,\Delta\iota[\tau]=\Delta\tilde{\iota}_l}
\mathcal Q_S[\tau].
\label{eq:grouped-weights}
\end{equation}
The grouped distribution $\{(\Delta\tilde{\iota}_l,\tilde{\mathcal Q}_l)\}_{l=1}^{m}$ therefore preserves the total weight. For this grouped distribution, Horv\'ath's criterion gives the necessary and sufficient conditions for the discrete Jensen-Steffensen inequality \cite{Horvath2024}.

\begin{theorem}
\label{thm:horvath-criterion}
The Jensen-Steffensen inequality holds for the distribution $\{(\Delta\tilde{\iota}_l,\tilde{\mathcal Q}_l)\}_{l=1}^{m}$ and every continuous convex function if and only if
\begin{subequations}
\label{eq:horvath-condition}   
\begin{gather}
\label{eq:horvath-left-condition}
\sum_{l=1}^{\ell}\tilde{\mathcal{Q}}_l
\bigl(\Delta\tilde{\iota}_l-\Delta\tilde{\iota}_{\ell+1}\bigr)
\ge 0,
\\
\label{eq:horvath-right-condition}
\sum_{l=m+1-\ell}^{m}\tilde{\mathcal{Q}}_l
\bigl(\Delta\tilde{\iota}_{m-\ell}-\Delta\tilde{\iota}_l\bigr)
\ge 0,
\end{gather}
\end{subequations}
for $\ell=1,\dots,m-1$.
\end{theorem}

In the present problem, the grouped weights are normalized as $\sum_{l=1}^{m}\tilde{\mathcal Q}_l=1$. Note that the conditions must hold for the Jensen-Steffensen inequality to be valid for every convex function. Conditions \eqref{eq:horvath-left-condition} and \eqref{eq:horvath-right-condition} test the signed distribution from the two ends of the ordered support and replace the nonnegative partial-sum conditions used in the positive-weight setting. Grouping equal support values is essential here because only the total weight at each distinct value of $\Delta\iota[\gamma]$ matters.

Among the conditions in Eq.~\eqref{eq:horvath-condition}, the right-end condition with $\ell=1$ is particularly useful and is given by
\begin{equation}
   \tilde{\mathcal{Q}}_m
   \bigl(\Delta\tilde{\iota}_{m-1}-\Delta\tilde{\iota}_m\bigr)
   \ge 0 .
\end{equation}
If the smallest support point $\Delta\tilde{\iota}_m$ is unique, namely $\Delta\tilde{\iota}_{m-1}>\Delta\tilde{\iota}_m$, this condition is equivalent to $\tilde{\mathcal Q}_m\geq0$. We therefore have the following test. 
\begin{corollary}
\label{Q_m}
Given the smallest grouped support point $\Delta\tilde{\iota}_m$ with grouped weight $\tilde{\mathcal{Q}}_m$, if $\tilde{\mathcal{Q}}_m<0$, then Horv\'ath's conditions fail.
\end{corollary}

Failure of Horv\'ath's criterion means that the Jensen--Steffensen inequality fails for at least one continuous convex function. The exponential relation in Eq.~\eqref{eq:valid_jensen_inequality} remains valid under these positive-spectrum assumptions.

\subsection{Symmetry-assisted grouping for an initial Bell state\label{Calculation of the Criterion}}

We first give a simple example in which negative KD quasiprobabilities do not violate Horv\'ath's criterion. Consider the initial Bell state $\rho_S=\ket{\Phi^+}\!\bra{\Phi^+}$ with the correlation coefficients $c_1=-c_2=c_3=1$, under local amplitude damping with $0<\mu<1$. The only nonzero initial Bell-state eigenvalue is $p_{k=\Phi^+}=1$. The remaining zero-eigenvalue initial states have zero KD weight and are omitted from the nonzero trajectory support. The final local and global eigenvalues are given by Eqs. \eqref{eq:decay_local_eigenvalues} and \eqref{eq:decay_global_eigenvalues}, respectively. The distribution of $\Delta\iota[\gamma]$ is degenerate because
\begin{multline}
\Delta\iota[\Phi^+,00,\Omega^\pm,00]
=\Delta\iota[\Phi^+,11,\Omega^\pm,00]\\
=\ln\left(\frac{2(1+\mu)^2}{1-\mu+\mu^2\pm\sqrt{\mu^2+\bar\mu^2}}\right),
\end{multline}
where $\bar\mu=1-\mu$ and $\mu$ is the damping probability in Eq. \eqref{eq:kraus_amplitude_damping}.

For $v=2>0$, Eq.~\eqref{eq:Q_S_v>0} shows that the only negative trajectory with nonzero initial weight is $[\Phi^+,00,\Omega^-,00]$, which is degenerate with the trajectory $[\Phi^+,11,\Omega^-,00]$. Their quasiprobability weights are
\begin{subequations}
\label{eq:bell-reference-negative-pair}
\begin{align}
\mathcal Q_S[\Phi^+,00,\Omega^-,00]
&=\frac{\sqrt{\mu^2+\bar\mu^2}-\mu-\bar\mu^2}
{4\sqrt{\mu^2+\bar\mu^2}}<0,\\
\mathcal Q_S[\Phi^+,11,\Omega^-,00]
&=\frac{\mu^2\bigl(\sqrt{\mu^2+\bar\mu^2}-\mu\bigr)}
{4\sqrt{\mu^2+\bar\mu^2}}>0.
\end{align}
\end{subequations}
Thus, the negative contribution cannot be assessed independently of its degenerate positive partner. Summing the KD weights of the degenerate trajectories $[\Phi^+,00,\Omega^-,00]$ and $[\Phi^+,11,\Omega^-,00]$ gives
\begin{multline}
\mathcal Q_S[\Phi^+,00,\Omega^-,00]
+\mathcal Q_S[\Phi^+,11,\Omega^-,00]\\
=\frac{\bar\mu^2\bigl(1-\mu+\mu^2-\sqrt{\mu^2+\bar\mu^2}\bigr)}
{4\sqrt{\mu^2+\bar\mu^2}\bigl(\sqrt{\mu^2+\bar\mu^2}+\mu\bigr)}.
\end{multline}
Since $(1-\mu+\mu^2)^2-(\mu^2+\bar\mu^2)=\mu^2\bar\mu^2>0$, the summed weight above is positive. Moreover, all other trajectory weights are positive, so the grouped weights form an ordinary probability distribution. Consequently, the Jensen-Steffensen inequality holds for every continuous convex function, and hence Horv\'ath's criterion is satisfied. This example demonstrates that microscopic negativity alone is insufficient to establish a violation. 

\subsection{Horv\'ath's criterion for the most negative quasiprobability trajectory\label{scend Calculation of the Criterion}}

The previous example with the simple symmetric parameter choice shows that negativity can be absorbed by grouping at a degenerate support value. To test whether Horv\'ath's criterion fails for the KD quasiprobability weights, the most natural candidate is the most negative trajectory obtained in Theorem~\ref{thm:global-most-negative}, since a more negative weight is expected to bring the Jensen-Steffensen inequality closer to violation.

Recall that the most negative quasiprobability is obtained for $c_1=-c_2=v_{\mathcal Q_{\min}}/2$ and $c_3=1$, where $v_{\mathcal Q_{\min}}$ is given by Eq.~\eqref{eq:rstar-definition}. For this family, the initial populations $p_{k=\Psi^\pm}$ vanish, so trajectories originating from these zero-population states are omitted from the nonzero trajectory support. The analytic support values are listed in Eq.~\eqref{eq:global-minimizer-support-values} of Appendix~\ref{app:global_minimizer_support}. We find that the distribution for the most negative quasiprobability family with $0<\mu<\mu_c$ has a unique smallest support value, namely $\Delta\iota[\Phi^-,11,\Omega^+,11]$. The quasiprobability weight associated with this unique smallest support value is given in Eq.~\eqref{eq:Q_S_v>0_2}. As shown in Sec.~\ref{amplitude damping channel}, $\mathcal Q_S[\Phi^-,11,\Omega^+,11]$ is negative as long as $v>0$. Corollary~\ref{Q_m} therefore implies that Horv\'ath's criterion fails. This conclusion is also consistent with Fig.~\ref{fig:result2}, where the orange point associated with this negative weight occurs at the smallest value of $\Delta\iota[\gamma]$.

This violation does not mean that the trajectory responsible for the violation of Horv\'ath's criterion is itself the most negative reduced quasiprobability trajectory. The most negative trajectory is $\mathcal Q_S[\Phi^+,00,\Omega^-,00]$, as stated in Theorem~\ref{thm:global-most-negative}. Horv\'ath's endpoint test instead depends on the ordering of $\Delta\iota[\gamma]$ and on the grouped weight at its smallest support value. For $\mu_c\leq\mu<1$, the optimization gives $v_{\mathcal Q_{\min}}=2$, for which $p_{k=\Phi^-}=0$. This boundary family coincides with the initial Bell-state example considered above, for which the negative contribution is grouped with its degenerate positive partner and Horv\'ath's criterion is satisfied. Thus, the endpoint-violation argument applies only in the range $0<\mu<\mu_c$.

\subsection{Horv\'ath's criterion for a general negative quasiprobability trajectory}

\label{third Calculation of the Criterion}

In the previous subsection, we showed that Horv\'ath's criterion fails for the parameter family yielding the most negative quasiprobability trajectory. We now show that a violation of Horv\'ath's criterion can also appear in other regimes through the endpoint test in Corollary~\ref{Q_m}. 

Consider a two-parameter family of initial states given by 
\begin{equation}
\label{eq:structured-slice}
c_1=c_3=-r,\quad c_2=-(r-\delta),
\end{equation}
where $r>0$, $\delta>0$, and $r+\delta<1$ ensure that the state lies within the Bell-state tetrahedron. The damping probability $\mu$ in Eq.~\eqref{eq:kraus_amplitude_damping} lies in the range $0<\mu<1$. To apply the endpoint test of Horv\'ath's criterion, we first identify the unique smallest support value of $\Delta\iota[\gamma]$.

\begin{lemma}
\label{lem:structured-minimum-condition}
Define
\begin{equation}
\label{eq:Gmu-definition}
G_\mu(r,\delta)=(1+\mu)B+\bar\mu^2\delta+2\mu(1+\mu)-4\bar\mu^2r,
\end{equation}
where the parameter $B$ is defined in Eq.~\eqref{eq:A_and_B}. Then $\Delta\iota[\Phi^+,11,\Omega^+,11]$ is the unique smallest support value if and only if $G_\mu(r,\delta)>0$.
\end{lemma}

The function $G_\mu(r,\delta)$ arises from comparing $\Delta\iota[\Phi^+,11,\Omega^+,11]$ with the only nontrivial competing support value $\Delta\iota[\Phi^+,11,\Psi^-,01]=\Delta\iota[\Phi^+,11,\Psi^-,10]$, and $G_\mu(r,\delta)>0$ is equivalent to the former being the smaller of the two. A detailed proof is provided in Appendix~\ref{app:violation_Horvath}. Equipped with the smallest support value identified in the lemma above, we establish the violation of Horv\'ath's criterion stated in the following theorem. 

\begin{theorem}
\label{thm:structured-slice-violation}
For the structured parameter family defined in Eq.~\eqref{eq:structured-slice}, if $G_\mu(r,\delta)>0$, then the unique smallest support value $\Delta\iota[\Phi^+,11,\Omega^+,11]$ has the weight
\begin{equation}
\label{eq:structured-Qstar}
\mathcal Q_S[\Phi^+,11,\Omega^+,11] = \frac{1-r-\delta}{8}\,
\frac{\bar\mu^2\delta\bigl(\bar\mu^2\delta-2\mu-B\bigr)}
     {\bar\mu^2\delta^2+(2\mu+B)^2},
\end{equation}
which is negative and therefore violates Horv\'ath's criterion. 
\end{theorem}
The proof is provided in Appendix~\ref{app:violation_Horvath}.
Specifically, it evaluates the transition quasiprobability from the components of $\ket{\Omega^+}$, establishes the negativity of Eq.~\eqref{eq:structured-Qstar}, and combines this result with the uniqueness condition $G_\mu(r,\delta)>0$ to apply the endpoint test in Corollary~\ref{Q_m}.

The condition $G_\mu(r,\delta)>0$ is equivalent to $0<r<r_c(\delta,\mu)$, where
\begin{multline}
r_c(\delta,\mu)=\frac{1}{4\bar\mu^2}\Big((1+\mu)\sqrt{\bar\mu^2\delta^2+4\mu^2}\\
+\bar\mu^2\delta+2\mu(1+\mu)\Big).
\end{multline}
The endpoint-violation region can therefore be written as $0<r<\min\left(1-\delta,r_c(\delta,\mu)\right)$. 
Evaluating $G_\mu(r,\delta)$ on the physical boundary $r=1-\delta$ gives
\begin{multline}
G_\mu(1-\delta,\delta) =(1+\mu)\sqrt{\bar\mu^2\delta^2+4\mu^2}
 +5\bar\mu^2\delta\\
 +2\mu(1+\mu)-4\bar\mu^2.
\end{multline}
This boundary expression is strictly increasing in $\delta$. Since $G_\mu(1,0)=4(3\mu-1)$, the endpoint test of Horv\'ath's criterion is violated throughout the entire physical region whenever $\mu\geq1/3$. 

For $0<\mu<1/3$, there is a unique threshold
\begin{equation}
\label{eq:delta0-definition}
\delta_0(\mu)=
\frac{10-25\mu+5\mu^2
-(1+\mu)\sqrt{25\mu^2-12\mu+4}}
{2(6\mu^2-13\mu+6)}
\end{equation}
such that the endpoint-violation condition is $0<r<r_c(\delta,\mu)$ for $0<\delta<\delta_0(\mu)$. For example, at $\mu=0.2$, Eq.~\eqref{eq:delta0-definition} gives $\delta_0(0.2)\simeq0.4485$. Thus, for $0<\delta<\delta_0(0.2)$, the endpoint-violation region is $0<r<r_c(\delta,0.2)$. For $\delta_0(0.2)\leq\delta<1$, the criterion is violated throughout the entire physical region $0<r<1-\delta$. The boundary between the two regimes is the curve $G_\mu(r,\delta)=0$, as illustrated in Fig.~\ref{fig:violation}.

The present counterexamples correct the inference following Eq.~(45) of Ref.~\cite{Zhang2026Multipartite} that the exponential Jensen-like relation forces the real quasiprobability weights to satisfy Horv\'ath's criterion. The FT and the exponential inequality remain valid, but their validity does not imply the Jensen--Steffensen inequality for every continuous convex function. Experimentally, state and process tomography can reconstruct Eq.~\eqref{eq:Q_S} with the specified rank-one projectors. Alternatively, a two-qubit adaptation of ancilla-assisted sequential weak-value measurements would probe $\Pi_s$ before and $\Pi_{s'}$ after the channel, condition on $k$, and postselect $k'$, with preparation and postselection probabilities restored \cite{Kim2018SequentialProcess}. Measured spectra determine $\Delta\iota$. Grouping equal values then permits testing Eq.~\eqref{eq:horvath-condition}. A negative smallest grouped weight suffices provided its negativity and support ordering are determined.

\begin{figure}
    \centering
    \includegraphics[width=1\linewidth]{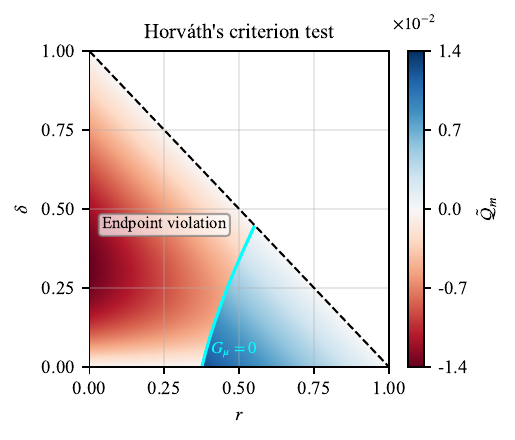}
    \caption{Grouped quasiprobability weight $\tilde{\mathcal Q}_m$ at the smallest support value of $\Delta\iota[\gamma]$ in the $(r,\delta)$ plane for the structured parameters in Eq.~\eqref{eq:structured-slice} with $\mu=0.2$. Red and blue denote negative and positive values, respectively. The cyan curve is the boundary $G_\mu(r,\delta)=0$. The analytic endpoint-violation region is defined by $G_\mu(r,\delta)>0$. Only the physical region $0<r<1-\delta$ is considered.}
    \label{fig:violation}
\end{figure}

\section{\label{sec:conclusion} Conclusions and outlook}

Using two-qubit Bell-diagonal states subject to three representative decoherence channels, we have systematically investigated negativity in the KD quasiprobability trajectories that describe the statistics of quantum mutual information dissipation~\cite{Zhang2026Multipartite}. For the fixed computational-basis local projector, all physically allowed initial Bell-diagonal states yield nonnegative trajectory weights under local dephasing and depolarizing dynamics, whereas local amplitude damping produces negative trajectories whenever $v=c_1-c_2\neq0$. Thus, the trajectory negativity is not determined by entanglement or CHSH nonlocality. 

Optimizing the initial-state parameters and trajectory labels with a fixed local projector, we have analytically determined the most negative transition quasiprobability $q[\gamma]$ and reduced trajectory weight $\mathcal Q_S[\gamma]=p_kq[\gamma]$ in the amplitude-damping dynamics, obtaining the parameter-independent bounds $q[\gamma]>(1-\sqrt{2})/4$ and $\mathcal Q_S[\gamma]>(1-\sqrt{2})/8$. For these channels, the channel-dependent analysis in Sec.~\ref{Interference decomposition of the quasiprobability} identifies a nonnegative diagonal reference and bounds the interference contribution, given by Eqs.~\eqref{eq:population-reference} and~\eqref{eq:channel-interference-bound}, relating trajectory negativity to the channel matrix elements, final eigenvector components, and their relative phases.

A further distinction emerges in the relation between quasiprobability negativity and convexity. Although the standard Jensen's inequality cannot be assumed for distributions with negative weights, the Jensen-like relation $\langle e^{-\Delta\iota}\rangle_{\mathcal Q_S}\geq e^{-\langle\Delta\iota\rangle_{\mathcal Q_S}}$ nevertheless follows from the integral FT and the quantum data-processing inequality. We prove that the real signed trajectory distributions in the amplitude-damping examples violate Horv\'ath's necessary and sufficient conditions for the Jensen-Steffensen inequality to hold for every continuous convex function~\cite{Horvath2024}. Thus, once Horv\'ath's criterion is violated, Jensen-type inequalities for other information measures are no longer guaranteed by the quasiprobability trajectory and must be examined separately. Generalized entropies, including Tsallis and linear entropies, provide natural candidates for such future investigations~\cite{Tsallis1988,ManfrediFeix2000,Rungta2001}.

A natural extension of our work is to test the interference decomposition and the optimized lower bounds on KD quasiprobability trajectory weights for general bipartite and multipartite states, as well as for correlated, non-Markovian, or collective open-system dynamics~\cite{Dicke1954,Lehmberg1970,Pechukas1994,BreuerLainePiilo2009}. Another open problem is to extend this trajectory perspective to quantum R\'enyi divergences~\cite{Petz1986,FrankLieb2013,Leditzky2017}. It would be valuable to determine whether the monotonicity of R\'enyi divergences has a trajectory-level counterpart, namely, whether the corresponding R\'enyi data-processing inequality holds precisely when the KD quasiprobability distribution satisfies a generalized Jensen-Steffensen condition, and whether failure of this inequality can be traced to the loss of this convexity compatibility. Finally, an important goal is to determine whether an appropriate generalized quantum Jensen's inequality can be used to derive the quantum data-processing inequality from the KD quasiprobability FT even when Horv\'ath's criterion is violated. Such a result would bridge the gap between classical and quantum stochastic information dynamics.

\begin{acknowledgments}
This work was supported by the NSFC (Grants No. 12675026, No.12275215, and No.12574300), Scientific Research Program Funded by Education Department of Shaanxi Provincial Government (Program No. 24JP186 and No. 25JP183), and the Youth Innovation Team of Shaanxi Universities. KZ was supported by the China Postdoctoral Science Foundation under Grant Number 2025M773421, and Shaanxi Province Postdoctoral Science Foundation under Grant Number 2025BSHYDZZ017.
\end{acknowledgments}

\appendix

\section{All nonzero quasiprobability trajectories under the depolarizing channel}

\label{app:depolarizing_quasiprobability}

We list all nonzero reduced KD quasiprobability trajectories $\mathcal Q_S[\gamma]$ generated by the depolarizing channel with the fixed computational-basis local projector below. The transition KD quasiprobabilities $q[\gamma]\in\{\kappa_1,\kappa_2,\kappa_3,\kappa_4\}$ are given by Eq.~\eqref{eq:kappa_parameter}. The initial Bell-state eigenvalues $p_{k=\Phi^\pm}$ and $p_{k=\Psi^\pm}$ are given by Eq.~\eqref{eq:pk}. For $k=\Phi^+$ and $s=00$,
\begin{subequations}
\begin{align}
\mathcal{Q}_S[\Phi^+,00,\Phi^+,00]&=\kappa_1\,p_{k=\Phi^+}, \\
\mathcal{Q}_S[\Phi^+,00,\Phi^-,00]&=\kappa_2\,p_{k=\Phi^+}, \\
\mathcal{Q}_S[\Phi^+,00,\Phi^+,11]&=\kappa_3\,p_{k=\Phi^+}, \\
\mathcal{Q}_S[\Phi^+,00,\Phi^-,11]&=\kappa_3\,p_{k=\Phi^+}, \\
\mathcal{Q}_S[\Phi^+,00,\Psi^+,01]&=\kappa_4\,p_{k=\Phi^+}, \\
\mathcal{Q}_S[\Phi^+,00,\Psi^+,10]&=\kappa_4\,p_{k=\Phi^+}, \\
\mathcal{Q}_S[\Phi^+,00,\Psi^-,01]&=\kappa_4\,p_{k=\Phi^+}, \\
\mathcal{Q}_S[\Phi^+,00,\Psi^-,10]&=\kappa_4\,p_{k=\Phi^+}.
\end{align}
\end{subequations}
For $k=\Phi^+$ and $s=11$,
\begin{subequations}
\begin{align}
\mathcal{Q}_S[\Phi^+,11,\Phi^+,00]&=\kappa_3\,p_{k=\Phi^+}, \\
\mathcal{Q}_S[\Phi^+,11,\Phi^-,00]&=\kappa_3\,p_{k=\Phi^+}, \\
\mathcal{Q}_S[\Phi^+,11,\Phi^+,11]&=\kappa_1\,p_{k=\Phi^+}, \\
\mathcal{Q}_S[\Phi^+,11,\Phi^-,11]&=\kappa_2\,p_{k=\Phi^+}, \\
\mathcal{Q}_S[\Phi^+,11,\Psi^+,01]&=\kappa_4\,p_{k=\Phi^+}, \\
\mathcal{Q}_S[\Phi^+,11,\Psi^+,10]&=\kappa_4\,p_{k=\Phi^+}, \\
\mathcal{Q}_S[\Phi^+,11,\Psi^-,01]&=\kappa_4\,p_{k=\Phi^+}, \\
\mathcal{Q}_S[\Phi^+,11,\Psi^-,10]&=\kappa_4\,p_{k=\Phi^+}.
\end{align}
\end{subequations}
For $k=\Phi^-$ and $s=00$,
\begin{subequations}
\begin{align}
\mathcal{Q}_S[\Phi^-,00,\Phi^+,00]&=\kappa_2\,p_{k=\Phi^-}, \\
\mathcal{Q}_S[\Phi^-,00,\Phi^-,00]&=\kappa_1\,p_{k=\Phi^-}, \\
\mathcal{Q}_S[\Phi^-,00,\Phi^+,11]&=\kappa_3\,p_{k=\Phi^-}, \\
\mathcal{Q}_S[\Phi^-,00,\Phi^-,11]&=\kappa_3\,p_{k=\Phi^-}, \\
\mathcal{Q}_S[\Phi^-,00,\Psi^+,01]&=\kappa_4\,p_{k=\Phi^-}, \\
\mathcal{Q}_S[\Phi^-,00,\Psi^+,10]&=\kappa_4\,p_{k=\Phi^-}, \\
\mathcal{Q}_S[\Phi^-,00,\Psi^-,01]&=\kappa_4\,p_{k=\Phi^-}, \\
\mathcal{Q}_S[\Phi^-,00,\Psi^-,10]&=\kappa_4\,p_{k=\Phi^-}.
\end{align}
\end{subequations}
For $k=\Phi^-$ and $s=11$,
\begin{subequations}
\begin{align}
\mathcal{Q}_S[\Phi^-,11,\Phi^+,00]&=\kappa_3\,p_{k=\Phi^-}, \\
\mathcal{Q}_S[\Phi^-,11,\Phi^-,00]&=\kappa_3\,p_{k=\Phi^-}, \\
\mathcal{Q}_S[\Phi^-,11,\Phi^+,11]&=\kappa_2\,p_{k=\Phi^-}, \\
\mathcal{Q}_S[\Phi^-,11,\Phi^-,11]&=\kappa_1\,p_{k=\Phi^-}, \\
\mathcal{Q}_S[\Phi^-,11,\Psi^+,01]&=\kappa_4\,p_{k=\Phi^-}, \\
\mathcal{Q}_S[\Phi^-,11,\Psi^+,10]&=\kappa_4\,p_{k=\Phi^-}, \\
\mathcal{Q}_S[\Phi^-,11,\Psi^-,01]&=\kappa_4\,p_{k=\Phi^-}, \\
\mathcal{Q}_S[\Phi^-,11,\Psi^-,10]&=\kappa_4\,p_{k=\Phi^-}.
\end{align}
\end{subequations}
For $k=\Psi^+$ and $s=01$,
\begin{subequations}
\begin{align}
\mathcal{Q}_S[\Psi^+,01,\Phi^+,00]&=\kappa_4\,p_{k=\Psi^+}, \\
\mathcal{Q}_S[\Psi^+,01,\Phi^+,11]&=\kappa_4\,p_{k=\Psi^+}, \\
\mathcal{Q}_S[\Psi^+,01,\Phi^-,00]&=\kappa_4\,p_{k=\Psi^+}, \\
\mathcal{Q}_S[\Psi^+,01,\Phi^-,11]&=\kappa_4\,p_{k=\Psi^+}, \\
\mathcal{Q}_S[\Psi^+,01,\Psi^+,01]&=\kappa_1\,p_{k=\Psi^+}, \\
\mathcal{Q}_S[\Psi^+,01,\Psi^-,01]&=\kappa_2\,p_{k=\Psi^+}, \\
\mathcal{Q}_S[\Psi^+,01,\Psi^+,10]&=\kappa_3\,p_{k=\Psi^+}, \\
\mathcal{Q}_S[\Psi^+,01,\Psi^-,10]&=\kappa_3\,p_{k=\Psi^+}.
\end{align}
\end{subequations}
For $k=\Psi^+$ and $s=10$,
\begin{subequations}
\begin{align}
\mathcal{Q}_S[\Psi^+,10,\Phi^+,00]&=\kappa_4\,p_{k=\Psi^+}, \\
\mathcal{Q}_S[\Psi^+,10,\Phi^+,11]&=\kappa_4\,p_{k=\Psi^+}, \\
\mathcal{Q}_S[\Psi^+,10,\Phi^-,00]&=\kappa_4\,p_{k=\Psi^+}, \\
\mathcal{Q}_S[\Psi^+,10,\Phi^-,11]&=\kappa_4\,p_{k=\Psi^+}, \\
\mathcal{Q}_S[\Psi^+,10,\Psi^+,01]&=\kappa_3\,p_{k=\Psi^+}, \\
\mathcal{Q}_S[\Psi^+,10,\Psi^-,01]&=\kappa_3\,p_{k=\Psi^+}, \\
\mathcal{Q}_S[\Psi^+,10,\Psi^+,10]&=\kappa_1\,p_{k=\Psi^+}, \\
\mathcal{Q}_S[\Psi^+,10,\Psi^-,10]&=\kappa_2\,p_{k=\Psi^+}.
\end{align}
\end{subequations}
For $k=\Psi^-$ and $s=01$,
\begin{subequations}
\begin{align}
\mathcal{Q}_S[\Psi^-,01,\Phi^+,00]&=\kappa_4\,p_{k=\Psi^-}, \\
\mathcal{Q}_S[\Psi^-,01,\Phi^+,11]&=\kappa_4\,p_{k=\Psi^-}, \\
\mathcal{Q}_S[\Psi^-,01,\Phi^-,00]&=\kappa_4\,p_{k=\Psi^-}, \\
\mathcal{Q}_S[\Psi^-,01,\Phi^-,11]&=\kappa_4\,p_{k=\Psi^-}, \\
\mathcal{Q}_S[\Psi^-,01,\Psi^+,01]&=\kappa_2\,p_{k=\Psi^-}, \\
\mathcal{Q}_S[\Psi^-,01,\Psi^-,01]&=\kappa_1\,p_{k=\Psi^-}, \\
\mathcal{Q}_S[\Psi^-,01,\Psi^+,10]&=\kappa_3\,p_{k=\Psi^-}, \\
\mathcal{Q}_S[\Psi^-,01,\Psi^-,10]&=\kappa_3\,p_{k=\Psi^-}.
\end{align}
\end{subequations}
For $k=\Psi^-$ and $s=10$,
\begin{subequations}
\begin{align}
\mathcal{Q}_S[\Psi^-,10,\Phi^+,00]&=\kappa_4\,p_{k=\Psi^-}, \\
\mathcal{Q}_S[\Psi^-,10,\Phi^+,11]&=\kappa_4\,p_{k=\Psi^-}, \\
\mathcal{Q}_S[\Psi^-,10,\Phi^-,00]&=\kappa_4\,p_{k=\Psi^-}, \\
\mathcal{Q}_S[\Psi^-,10,\Phi^-,11]&=\kappa_4\,p_{k=\Psi^-}, \\
\mathcal{Q}_S[\Psi^-,10,\Psi^+,01]&=\kappa_3\,p_{k=\Psi^-}, \\
\mathcal{Q}_S[\Psi^-,10,\Psi^-,01]&=\kappa_3\,p_{k=\Psi^-}, \\
\mathcal{Q}_S[\Psi^-,10,\Psi^+,10]&=\kappa_2\,p_{k=\Psi^-}, \\
\mathcal{Q}_S[\Psi^-,10,\Psi^-,10]&=\kappa_1\,p_{k=\Psi^-}.
\end{align}
\end{subequations}

\section{All nonzero quasiprobability trajectories under the amplitude-damping channel}

\label{app:decay_quasiprobability}

We list all nonzero reduced KD quasiprobability trajectories $\mathcal Q_S[\gamma]$ generated by the amplitude-damping channel with the fixed computational-basis local projector below. The coefficients $\omega^\pm_{00}$ and $\omega^\pm_{11}$ are defined in Eq.~\eqref{eq:eigenvectors}. The initial Bell-state eigenvalues $p_{k=\Phi^\pm}$ and $p_{k=\Psi^\pm}$ are given by Eq.~\eqref{eq:pk}. For $k=\Phi^+$ and $s=00$,
\begin{subequations}
\begin{align}
\mathcal{Q}_S[\Phi^+,00,\Omega^+,00]
&=\frac12\left(\bigl(\omega^+_{00}\bigr)^2+\bar{\mu}\omega^+_{00}\omega^+_{11}\right)p_{k=\Phi^+},\\
\mathcal{Q}_S[\Phi^+,00,\Omega^-,00]
&=\frac12\left(\bigl(\omega^-_{00}\bigr)^2+\bar{\mu}\omega^-_{00}\omega^-_{11}\right)p_{k=\Phi^+}.
\end{align}
\end{subequations}
For $k=\Phi^-$ and $s=00$,
\begin{subequations}
\begin{align}
\mathcal{Q}_S[\Phi^-,00,\Omega^+,00]
&=\frac12\left(\bigl(\omega^+_{00}\bigr)^2-\bar{\mu}\omega^+_{00}\omega^+_{11}\right)p_{k=\Phi^-},\\
\mathcal{Q}_S[\Phi^-,00,\Omega^-,00]
&=\frac12\left(\bigl(\omega^-_{00}\bigr)^2-\bar{\mu}\omega^-_{00}\omega^-_{11}\right)p_{k=\Phi^-}.
\end{align}
\end{subequations}
For $k=\Phi^+$ and $s=11$,
\begin{subequations}
\begin{align}
\mathcal{Q}_S[\Phi^+,11,\Omega^+,00]
&=\frac{\mu^2}{2}\bigl(\omega^+_{00}\bigr)^2p_{k=\Phi^+},\\
\mathcal{Q}_S[\Phi^+,11,\Omega^+,11]
&=\frac12\left(\bar{\mu}\omega^+_{00}\omega^+_{11}+\bar{\mu}^2\bigl(\omega^+_{11}\bigr)^2\right)p_{k=\Phi^+},\\
\mathcal{Q}_S[\Phi^+,11,\Omega^-,00]
&=\frac{\mu^2}{2}\bigl(\omega^-_{00}\bigr)^2p_{k=\Phi^+},\\
\mathcal{Q}_S[\Phi^+,11,\Omega^-,11]
&=\frac12\left(\bar{\mu}\omega^-_{00}\omega^-_{11}+\bar{\mu}^2\bigl(\omega^-_{11}\bigr)^2\right)p_{k=\Phi^+},\\
\mathcal{Q}_S[\Phi^+,11,\Psi^+,01]
&=\frac{\bar{\mu}\mu}{4}\,p_{k=\Phi^+},\\
\mathcal{Q}_S[\Phi^+,11,\Psi^+,10]
&=\frac{\bar{\mu}\mu}{4}\,p_{k=\Phi^+},\\
\mathcal{Q}_S[\Phi^+,11,\Psi^-,01]
&=\frac{\bar{\mu}\mu}{4}\,p_{k=\Phi^+},\\
\mathcal{Q}_S[\Phi^+,11,\Psi^-,10]
&=\frac{\bar{\mu}\mu}{4}\,p_{k=\Phi^+}.
\end{align}
\end{subequations}
For $k=\Phi^-$ and $s=11$,
\begin{subequations}
\begin{align}
\mathcal{Q}_S[\Phi^-,11,\Omega^+,00]
&=\frac{\mu^2}{2}\bigl(\omega^+_{00}\bigr)^2p_{k=\Phi^-},\\
\mathcal{Q}_S[\Phi^-,11,\Omega^+,11]
&=\frac12\left(-\bar{\mu}\omega^+_{00}\omega^+_{11}+\bar{\mu}^2\bigl(\omega^+_{11}\bigr)^2\right)p_{k=\Phi^-},\\
\mathcal{Q}_S[\Phi^-,11,\Omega^-,00]
&=\frac{\mu^2}{2}\bigl(\omega^-_{00}\bigr)^2p_{k=\Phi^-},\\
\mathcal{Q}_S[\Phi^-,11,\Omega^-,11]
&=\frac12\left(-\bar{\mu}\omega^-_{00}\omega^-_{11}+\bar{\mu}^2\bigl(\omega^-_{11}\bigr)^2\right)p_{k=\Phi^-},\\
\mathcal{Q}_S[\Phi^-,11,\Psi^+,01]
&=\frac{\bar{\mu}\mu}{4}\,p_{k=\Phi^-},\\
\mathcal{Q}_S[\Phi^-,11,\Psi^+,10]
&=\frac{\bar{\mu}\mu}{4}\,p_{k=\Phi^-},\\
\mathcal{Q}_S[\Phi^-,11,\Psi^-,01]
&=\frac{\bar{\mu}\mu}{4}\,p_{k=\Phi^-},\\
\mathcal{Q}_S[\Phi^-,11,\Psi^-,10]
&=\frac{\bar{\mu}\mu}{4}\,p_{k=\Phi^-}.
\end{align}
\end{subequations}
For $k=\Psi^+$ and $s=01$,
\begin{subequations}
\begin{align}
\mathcal{Q}_S[\Psi^+,01,\Omega^+,00]
&=\frac{\mu}{2}\bigl(\omega^+_{00}\bigr)^2p_{k=\Psi^+},\\
\mathcal{Q}_S[\Psi^+,01,\Omega^-,00]
&=\frac{\mu}{2}\bigl(\omega^-_{00}\bigr)^2p_{k=\Psi^+},\\
\mathcal{Q}_S[\Psi^+,01,\Psi^+,01]
&=\frac{\bar{\mu}}{2}\,p_{k=\Psi^+}.
\end{align}
\end{subequations}
For $k=\Psi^+$ and $s=10$,
\begin{subequations}
\begin{align}
\mathcal{Q}_S[\Psi^+,10,\Omega^+,00]
&=\frac{\mu}{2}\bigl(\omega^+_{00}\bigr)^2p_{k=\Psi^+},\\
\mathcal{Q}_S[\Psi^+,10,\Omega^-,00]
&=\frac{\mu}{2}\bigl(\omega^-_{00}\bigr)^2p_{k=\Psi^+},\\
\mathcal{Q}_S[\Psi^+,10,\Psi^+,10]
&=\frac{\bar{\mu}}{2}\,p_{k=\Psi^+}.
\end{align}
\end{subequations}
For $k=\Psi^-$ and $s=01$,
\begin{subequations}
\begin{align}
\mathcal{Q}_S[\Psi^-,01,\Omega^+,00]
&=\frac{\mu}{2}\bigl(\omega^+_{00}\bigr)^2p_{k=\Psi^-},\\
\mathcal{Q}_S[\Psi^-,01,\Omega^-,00]
&=\frac{\mu}{2}\bigl(\omega^-_{00}\bigr)^2p_{k=\Psi^-},\\
\mathcal{Q}_S[\Psi^-,01,\Psi^-,01]
&=\frac{\bar{\mu}}{2}\,p_{k=\Psi^-}.
\end{align}
\end{subequations}
For $k=\Psi^-$ and $s=10$,
\begin{subequations}
\begin{align}
\mathcal{Q}_S[\Psi^-,10,\Omega^+,00]
&=\frac{\mu}{2}\bigl(\omega^+_{00}\bigr)^2p_{k=\Psi^-},\\
\mathcal{Q}_S[\Psi^-,10,\Omega^-,00]
&=\frac{\mu}{2}\bigl(\omega^-_{00}\bigr)^2p_{k=\Psi^-},\\
\mathcal{Q}_S[\Psi^-,10,\Psi^-,10]
&=\frac{\bar{\mu}}{2}\,p_{k=\Psi^-}.
\end{align}
\end{subequations}

\section{Dependence on the local rank-one projector}
\label{app:local-resolution}

In this appendix, we consider the quasiprobability trajectory weight with other local rank-one projectors. For each qubit, we use the same local basis at the initial and final times. All examples use the full-rank Bell-diagonal state $(c_1,c_2,c_3)=(1/5,0,2/5)$, whose populations, in the order $(\Phi^+,\Phi^-,\Psi^+,\Psi^-)$, are $(2/5,3/10,1/5,1/10)$. At the channel parameters used below, the final populations are $(29,27,13,11)/80$ for dephasing and $(23,21,19,17)/80$ for depolarization. Thus, all initial and final global eigenvalues are strictly positive and nondegenerate. Only the local states are degenerate.

For both qubits, we choose the rotated local basis
\begin{subequations}
\begin{align}
\label{eq:rotated-local-basis}
\ket{+_\theta}&=\cos\frac\theta2\ket0+\sin\frac\theta2\ket1,\\
\ket{-_\theta}&=-\sin\frac\theta2\ket0+\cos\frac\theta2\ket1.
\end{align}
\end{subequations}
For local dephasing of both qubits, using $+$ and $-$ to denote these local basis labels, we obtain
\begin{equation}
\label{eq:rotated-dephasing-negative}
q[\Psi^+,+-,\Psi^+,-+]
=-\frac{(1-\lambda)^2}{8}\sin^2\theta\cos^2\theta,
\end{equation}
where $\lambda\in[0,1]$ is the dephasing parameter defined in Eq. \eqref{eq:karus_dephasing}. At $\theta=0$, the local basis reduces to $\{\ket0,\ket1\}$. The weight in Eq.~\eqref{eq:rotated-dephasing-negative} then vanishes, whereas it is negative at $\theta=\pi/4$ for any fixed $\lambda<1$. Thus, even for a fixed initial state and channel, the KD trajectory weight and its sign depend on the chosen local basis.

Use the computational basis for qubit $S_1$ and Eq.~\eqref{eq:rotated-local-basis} for qubit $S_2$. For the trajectory $[\Phi^+,0-,\Psi^+,0+]$ under independent local depolarization of both qubits with the same parameter $\eta$, only the $\mathbb I\otimes\sigma^1$ and $\sigma^3\otimes\sigma^2$ error paths contribute, giving
\begin{multline}
\label{eq:rotated-depolarizing-negative}
q[\Phi^+,0-,\Psi^+,0+]=\frac{(1-\eta)\sin^2(\theta/2)}{32}\\
\times\left((1-\eta)-(1+3\eta)\cos\theta\right).
\end{multline}
Choose $\eta=1/2$ and $\cos\theta=3/5$, so that $\ket{+_\theta}=(2\ket0+\ket1)/\sqrt5$ and $\ket{-_\theta}=(-\ket0+2\ket1)/\sqrt5$. Equation~\eqref{eq:rotated-depolarizing-negative} then yields $q=-1/320$ and $\mathcal Q_S=-1/800$. Consequently, the unitary covariance of the depolarizing channel does not guarantee nonnegativity when only the local projector is changed.

The reality of the weights can also change. Choose local projectors $\Pi_{\pm}^{(j)}=(\mathbb I\pm\boldsymbol n_j\cdot\boldsymbol\sigma)/2$, with Bloch axes
\begin{equation}
\boldsymbol n_1=\frac{(1,0,1)}{\sqrt2},\qquad
\boldsymbol n_2=\frac{(0,1,1)}{\sqrt2}.
\end{equation}
For the trajectory $[\Phi^+,++,\Phi^-,++]$, direct substitution into Eq.~\eqref{eq:q-operator-sum} gives the following weight for independent local dephasing ($\lambda=1/2$): 
\begin{equation}
q[\Phi^+,++,\Phi^-,++]=\frac3{32}-\frac{3i\sqrt2}{128}.
\end{equation}
The corresponding reduced weight is obtained by multiplying by $p_{\Phi^+}=2/5$. These examples demonstrate dependence on the local rank-one projector even when the same basis is retained at both times. Their negative or nonreal contributions cancel upon summing the local labels at fixed $k,k'$, as required by Eq.~\eqref{eq:local-label-coarse-graining}.

\section{Proof of Theorem \ref{thm:global-transition-minimum}}

\label{app:proof_theorem_1}

Theorem~\ref{thm:global-transition-minimum} identifies the minimal (most negative) value of the transition KD quasiprobability $q[\gamma]$ under the amplitude-damping channel. Below is the detailed proof. 

\begin{proof}[Proof of Theorem~\ref{thm:global-transition-minimum}]
For $v>0$, only the following two transition quasiprobability trajectories are negative. All remaining nonzero transition quasiprobabilities are nonnegative.
\begin{subequations}
\begin{align}
q_+&=q[\Phi^+,00,\Omega^-,00]
=\frac{B(v)-2\mu-\bar{\mu}^2v}{4B(v)},\\
q_-&=q[\Phi^-,11,\Omega^+,11]
=\frac{\bar{\mu}^2\left(B(v)-2\mu-v\right)}{4B(v)},
\end{align}
\end{subequations}
where $B(v)=\sqrt{\bar{\mu}^2v^2+4\mu^2}$ and $0<v\leq2$. The first trajectory has
\begin{equation}
\frac{d q_+}{dv}
=\frac{\bar{\mu}^2\mu(v-2\mu)}{2B^3(v)},
\end{equation}
and is therefore minimized at $v=2\mu$. At this point, we have
\begin{equation}
\left.q_+\right|_{v=2\mu}-\left.q_-\right|_{v=2\mu}
=\frac{\left(\sqrt{1+\bar{\mu}^2}-1\right)
\left(1-\bar{\mu}^2\right)}{4\sqrt{1+\bar{\mu}^2}}>0.
\end{equation}
The minimum of the second trajectory cannot exceed its value at this point, whereas the first trajectory reaches its own minimum there and has a strictly larger value. The global minimum of the transition quasiprobability therefore lies on the second trajectory, namely $q_-=q[\Phi^-,11,\Omega^+,11]$. 

The derivative of the $q_-$ trajectory is
\begin{equation}
\frac{d q_-}{dv}
=-\frac{\bar{\mu}^2\mu\left(2\mu-\bar{\mu}^2v\right)}{2B^3(v)}.
\end{equation}
Thus, its stationary point is $v=2\mu/\bar{\mu}^2$. It belongs to the interval $0<v\leq2$ if and only if $2\mu/\bar{\mu}^2\leq2$, which is equivalent to $\mu^2-3\mu+1\geq0$. For $0<\mu<1$, this inequality holds precisely for $\mu\leq(3-\sqrt5)/2=\mu_q$, since its other root lies above unity. For $\mu>\mu_q$, $q_-$ decreases throughout the interval and reaches its minimum at $v=2$. Substituting the respective minimizing values of $v$ into $q_-$ yields Eq.~\eqref{eq:transition-q-min-value}.
\end{proof}

\section{Proof of Theorem \ref{thm:global-most-negative}}
\label{app:negative}

Theorem~\ref{thm:global-most-negative} identifies the minimal (most negative) value of the reduced KD quasiprobability $\mathcal Q_S[\gamma]$ under the amplitude-damping channel. Below is the detailed proof. 

\begin{proof}[Proof of Theorem~\ref{thm:global-most-negative}]

For $0<v\le2$, there are only two trajectories with negative weights, namely $\mathcal Q_S[\Phi^+,00,\Omega^-,00]$ and $\mathcal Q_S[\Phi^-,11,\Omega^+,11]$. Direct calculation gives
\begin{subequations}
\begin{align}
\label{eq:Q00}
\mathcal Q^+_{S}
&=\mathcal Q_S[\Phi^+,00,\Omega^-,00]=
\frac{(2+v)
\left(B(v)-2\mu-\bar{\mu}^2v\right)}
{16B(v)},
\\
\mathcal Q^-_{S}
&=\mathcal Q_S[\Phi^-,11,\Omega^+,11]=
\frac{\bar{\mu}^2(2-v)
\left(B(v)-2\mu-v\right)}
{16B(v)},
\end{align}
\end{subequations}
where the parameter $B(v)$ is defined in Eq. \eqref{eq:A_and_B}.

To compare the two trajectory weights, multiplying their difference by $16B(v)$ gives
\begin{equation}
16B(v)(\mathcal{Q}^+_{S}-\mathcal{Q}^-_{S})
= \bar{\mu}^2v^2\Bigl(\frac{2(1-\bar{\mu}^2)+v(1+\bar{\mu}^2)}{B(v)+2\mu}-2\Bigr).
\end{equation}
Therefore, $\mathcal Q^+_{S}<\mathcal Q^-_{S}$ is equivalent to 
\begin{equation}
   \frac{
2(1-\bar{\mu}^2)+v(1+\bar{\mu}^2)}
{B(v)+2\mu}
<2. 
\end{equation}
Using $1-\bar{\mu}^2=\mu(2-\mu),$
this inequality becomes
\begin{equation}
    2B(v)
+2\mu^2
-v(1+\bar{\mu}^2)
>0.
\end{equation}
Moreover, differentiating the left-hand side yields
\begin{equation}
    \frac{2\bar{\mu}^2v}{B(v)}
-(1+\bar{\mu}^2)
\le
2\bar{\mu}-(1+\bar{\mu}^2)
=
-(1-\bar{\mu})^2
\le0,
\end{equation}
so the left-hand side is monotonically decreasing on
\([0,2]\).
At $v=2$,
\begin{equation}
    2B(2)+2\mu^2-2(1+\bar{\mu}^2)
    =
4(\sqrt{\bar{\mu}^2+\mu^2}-\bar{\mu})
>0.
\end{equation}
Therefore, its minimum at $v=2$ is positive, and the inequality $\mathcal{Q}^+_{S}<\mathcal{Q}^-_{S}$ holds for every
\(0<\mu<1\) and
\(0<v\le2\).
Consequently, for \(v>0\), the minimum reduced KD quasiprobability is attained by $[\Phi^+,00,\Omega^-,00]$,
while the result for \(v<0\) follows immediately from the exchange
\(\Phi^+\leftrightarrow\Phi^-\).

Next we derive the minimizing value
$v_{\mathcal Q_{\min}}$
introduced in Eq.~\eqref{eq:rstar-definition}.  Consider the
one-variable objective function
\begin{equation}
F(v)
=
\mathcal{Q}_S[\Phi^+,00,\Omega^-,00],
\label{eq:app-F}
\end{equation}
with $0\le v\le2$. The derivative of the transition quasiprobability is
\begin{equation}
\frac{d}{dv}
q[\Phi^+,00,\Omega^-,00]
=
\frac{\bar{\mu}^2\mu(v-2\mu)}
{2B^3(v)}.
\label{eq:app-dq}
\end{equation}
Differentiating Eq.~\eqref{eq:app-F} gives
\begin{equation}
F'(v)
=
\frac1{16}
\left(
1
-
\frac{2\mu+\bar{\mu}^2v}
{B(v)}
-
\frac{
2\bar{\mu}^2\mu
(2+v)
(2\mu-v)}
{B^3(v)}
\right).
\label{eq:app-Fprime}
\end{equation}
For
$0<v<2\mu$,
Eq.~\eqref{eq:app-Fprime} can be rewritten as
\begin{equation}
16F'(v)
=
\frac{B(v)-2\mu-\bar{\mu}^2v}
{B(v)}
+
\frac{
2\bar{\mu}^2\mu
(2+v)
(v-2\mu)}
{B^3(v)}.
\label{eq:app-Fprime-alt}
\end{equation}
Since
$q[\Phi^+,00,\Omega^-,00]<0$,
the first term in Eq.~\eqref{eq:app-Fprime-alt} is negative.
The second term is also negative because
$v<2\mu$.
Hence $F'(v)<0$ for $0<v<2\mu$.

At
$v=2\mu$, we have
\begin{equation}
16F'(2\mu)
=
1-\sqrt{1+\bar{\mu}^2}
<
0.
\label{eq:app-middle}
\end{equation}
For
$v\ge2\mu$,
we introduce $t = \bar{\mu}v/2\mu$ and $R=\sqrt{1+t^2}$. The sign of
$dF/dt$
is determined by
\begin{equation}
G(t)
=
\mu\left(R^3-\bar{\mu}t^3\right)
+
\bar{\mu}(2\bar{\mu}-1)t
-
(\bar{\mu}^2-\bar{\mu}+1).
\label{eq:app-G}
\end{equation}
Differentiating $G(t)$ gives
\begin{equation}
G'(t)
=3\mu t(R-\bar{\mu}t)
+
\bar{\mu}(2\bar{\mu}-1).
\label{eq:app-Gprime}
\end{equation}
If
$\bar{\mu}\ge1/2$,
then the derivative in Eq.~\eqref{eq:app-Gprime} is positive. If
$\bar{\mu}<1/2$,
define $h(t)
=
t(R-\bar{\mu}t)$. Its derivative is
\begin{equation}
h'(t)
=
\frac{1+2t^2}{R}
-
2\bar{\mu}t
>
2\mu t
>
0,
\label{eq:app-hprime}
\end{equation}
so that $h(t)\ge h(\bar{\mu})$, which implies
\begin{equation}
3\mu t(R-\bar{\mu}t)
+
\bar{\mu}(2\bar{\mu}-1)>0.
\label{eq:app-Gprime-case2}
\end{equation}
Therefore,
$F'(v)$
has at most one zero in the interval
$(2\mu,2)$. Using Eq.~\eqref{eq:app-Fprime},
the stationary condition becomes
\begin{equation}
1
-
\frac{2\mu+\bar{\mu}^2v}
{B(v)}
-
\frac{
2\bar{\mu}^2\mu
(2+v)
(2\mu-v)}
{B^3(v)}
=
0.
\label{eq:app-stationary-explicit}
\end{equation}

The occurrence of an interior or boundary minimum is determined by $F'(2)=0$. Substituting
$v=2$
into Eq.~\eqref{eq:app-Fprime}
yields
\begin{equation}
16\mu^5
-
48\mu^4
+
68\mu^3
-
52\mu^2
+
21\mu
-
4
=
0.
\label{eq:app-critical-poly}
\end{equation}
This polynomial has a unique root in
$(0,1)$, given by $\mu_c\approx 0.82668$. Consequently,
the minimizer of
$F(v)$
is
\begin{equation}
v_{\mathcal Q_{\min}}
=
\begin{cases}
\nu_0,
&
0<\mu<\mu_c,
\\
2,
&
\mu_c\le\mu<1,
\end{cases}
\label{eq:app-rstar}
\end{equation}
where
$\nu_0$
is the unique solution of
Eq.~\eqref{eq:app-stationary-explicit}
in the interval
$(2\mu,2)$.

For each fixed $v$, the Bell-tetrahedron constraints give $p_{\Phi^+}\leq(2+v)/4$, with equality only at $c_3=1$ and $u=c_1+c_2=0$. Since $\mathcal Q^+_{S}=\mathcal Q_S[\Phi^+,00,\Omega^-,00]$ is negative, this choice gives the smallest complete weight at fixed $v$. Together with the comparison $\mathcal Q^+_{S}<\mathcal Q^-_{S}$ and the minimization of $F(v)=\mathcal Q^+_{S}$ above, these facts establish the minimum of $\mathcal Q_S[\gamma]$ for $v>0$. The equality conditions imply $c_1=-c_2=v_{\mathcal Q_{\min}}/2$. Substitution into $F(v)$ gives the first line of Eq.~\eqref{eq:global-Qmin}, while inserting $v_{\mathcal Q_{\min}}=2$ gives its second line. The $v<0$ case follows from the exchange $\Phi^+\leftrightarrow\Phi^-$. 
\end{proof}

\section{Stochastic mutual information changes with the most negative quasiprobability trajectory}
\label{app:global_minimizer_support}

Consider the family $c_1=-c_2=v/2$ and $c_3=1$, for which $u=0$ and the initial populations $p_{k=\Psi^\pm}$ vanish. Because the initial local marginals are maximally mixed, $p_{s_1}=p_{s_2}=1/2$, Eq.~\eqref{eq:stochastic-mutual information} reduces on every nonzero trajectory to
\begin{equation}
\Delta\iota[\gamma]
=\ln\!\left(
\frac{4p_kp_{s_1'}p_{s_2'}}{p_{k'}}
\right).
\label{eq:app-global-minimizer-support-rule}
\end{equation}
The nonzero trajectories listed in Appendix~\ref{app:decay_quasiprobability} consequently give the following support values. Several trajectories share the same value, so the displayed trajectory labels reduce to ten analytical expressions:
\begin{subequations}
\label{eq:global-minimizer-support-values}
\begin{align}
\Delta\iota[\Phi^+,00,\Omega^+,00]
&=\Delta\iota[\Phi^+,11,\Omega^+,00]
\nonumber\\
&=\ln\!\left(
\frac{(2+v)(1+\mu)^2}{A+B}
\right),\\[1mm]
\Delta\iota[\Phi^+,00,\Omega^-,00]
&=\Delta\iota[\Phi^+,11,\Omega^-,00]
\nonumber\\
&=\ln\!\left(
\frac{(2+v)(1+\mu)^2}{A-B}
\right),\\[1mm]
\Delta\iota[\Phi^-,00,\Omega^+,00]
&=\Delta\iota[\Phi^-,11,\Omega^+,00]
\nonumber\\
&=\ln\!\left(
\frac{(2-v)(1+\mu)^2}{A+B}
\right),\\[1mm]
\Delta\iota[\Phi^-,00,\Omega^-,00]
&=\Delta\iota[\Phi^-,11,\Omega^-,00]
\nonumber\\
&=\ln\!\left(
\frac{(2-v)(1+\mu)^2}{A-B}
\right),\\[1mm]
\Delta\iota[\Phi^+,11,\Omega^+,11]
&=\ln\!\left(
\frac{(2+v)\bar{\mu}^2}{A+B}
\right),\\
\Delta\iota[\Phi^+,11,\Omega^-,11]
&=\ln\!\left(
\frac{(2+v)\bar{\mu}^2}{A-B}
\right),\\[1mm]
\Delta\iota[\Phi^+,11,\Psi^\pm,01]
&=\Delta\iota[\Phi^+,11,\Psi^\pm,10]
\nonumber\\
&=\ln\!\left(
\frac{(2+v)(1-\mu^2)}{2\bar{\mu}\mu}
\right),\\[1mm]
\Delta\iota[\Phi^-,11,\Omega^+,11]
&=\ln\!\left(
\frac{(2-v)\bar{\mu}^2}{A+B}
\right),\\
\Delta\iota[\Phi^-,11,\Omega^-,11]
&=\ln\!\left(
\frac{(2-v)\bar{\mu}^2}{A-B}
\right),\\[1mm]
\Delta\iota[\Phi^-,11,\Psi^\pm,01]
&=\Delta\iota[\Phi^-,11,\Psi^\pm,10]
\nonumber\\
&=\ln\!\left(
\frac{(2-v)(1-\mu^2)}{2\bar{\mu}\mu}
\right).
\end{align}
\end{subequations}
Here $A$ and $B$ are defined in Eq.~\eqref{eq:A_and_B}. For the present family of the initial parameters, we have $A=2(1-\mu+\mu^2)$.

We now identify the smallest support value for the parameters $0<\mu<\mu_c$ and $v=v_{\mathcal Q_{\min}}\in(0,2)$. The two nonzero initial populations are 
\begin{equation}
p_{k=\Phi^-}=\frac{2-v}{4}<\frac{2+v}{4}=p_{k=\Phi^+}.
\end{equation}
The final local populations are
\begin{equation}
p_{s_j'=1}=\frac{\bar\mu}{2}<\frac{1+\mu}{2}=p_{s_j'=0},
\end{equation}
which satisfy $p_{s_j'=1}^2<p_{s_j'=0}p_{s_j'=1}<p_{s_j'=0}^2$. 
Moreover, $p_{k'=\Omega^+}$ is the largest final global eigenvalue. Indeed, $B>0$ gives $p_{k'=\Omega^+}>p_{k'=\Omega^-}$, while $p_{k'=\Psi^\pm}=\bar\mu\mu/2$ and
\begin{multline}
4\left(p_{k'=\Omega^+}-p_{k'=\Psi^\pm}\right)
=A+B-2\bar\mu\mu\\
>A-2\bar\mu\mu
=4\left(\mu-\frac12\right)^2+1>0.
\end{multline}
Since the logarithm in Eq.~\eqref{eq:app-global-minimizer-support-rule} is strictly increasing, its argument is minimized by choosing the smaller initial population $p_{k=\Phi^-}$, the smaller final local product $p_{s_1'=1}p_{s_2'=1}$, and the larger final global population $p_{k'=\Omega^+}$. The nonzero-trajectory list shows that the only trajectory realizing all three choices is $[\Phi^-,11,\Omega^+,11]$. Hence
\begin{equation}
\Delta\iota[\Phi^-,11,\Omega^+,11]
=\ln\!\left(
\frac{(2-v)\bar\mu^2}{A+B}
\right)
\end{equation}
is the unique smallest support value. At the boundary $\mu_c\leq\mu<1$, one has $v_{\mathcal Q_{\min}}=2$ and therefore $p_{k=\Phi^-}=0$. This trajectory is then omitted from the nonzero trajectory support, so the uniqueness statement above applies only to $0<\mu<\mu_c$.

\section{Proof of the violation of Horv\'ath's criterion under amplitude damping}

\label{app:violation_Horvath}

We present a detailed proof of the violation of Horv\'ath's criterion by first proving Lemma~\ref{lem:structured-minimum-condition}, which identifies the smallest support value of $\Delta\iota[\gamma]$ for the initial-state parameters in Eq.~\eqref{eq:structured-slice}. 

\begin{proof}[Proof of Lemma~\ref{lem:structured-minimum-condition}]
The parameters in Eq.~\eqref{eq:structured-slice} give the initial Bell-state eigenvalues
\begin{subequations}
\begin{align}
p_{k=\Phi^+}&=\frac{1-r-\delta}{4},\\
p_{k=\Psi^-}&=\frac{1+3r-\delta}{4}, \\
p_{k=\Phi^-}&=p_{k=\Psi^+}=\frac{1-r+\delta}{4}.
\end{align}
\end{subequations}
Thus \(p_{k=\Phi^+}<p_{k=\Phi^-}=p_{k=\Psi^+}\) and \(p_{k=\Psi^-}>p_{k=\Phi^+}\). The corresponding final global eigenvalues are
\begin{subequations}
\begin{align}
p_{k'=\Psi^+} &= \frac{\bar\mu}{4}\left((1+\mu)(1-r)+\delta\right),\\
p_{k'=\Psi^-} &= \frac{\bar\mu}{4}\left((1+\mu)+(3-\mu)r-\delta\right),\\
p_{k'=\Omega^\pm} &= \frac{1}{4}\bigl(1+\mu^2-\bar\mu^2r\pm B\bigr).
\end{align}
\end{subequations}
The final local eigenvalues are
\begin{equation}
p_{s'_j=0}=\frac{1+\mu}{2},\quad p_{s'_j=1}=\frac{\bar\mu}{2},
\end{equation}
for $j=1,2$. Consequently, \(p_{s'_1=1}p_{s'_2=1}=\bar\mu^2/4\) is the smallest final local product.

From $\Delta\iota[\gamma] = \ln\bigl(4p_k\,p_{s'_1}p_{s'_2}/p_{k'}\bigr)$, we compare the support values using
\begin{equation}
R[k,k',s']=\exp\Delta\iota[\gamma]
=\frac{4p_kp_{s'_1}p_{s'_2}}{p_{k'}},
\end{equation}
where $s'\in\{00,01,10,11\}$. Since the logarithm is strictly increasing, ordering the support values of $\Delta\iota[\gamma]$ is equivalent to ordering the corresponding values of $R[k,k',s']$.

We first minimize $R[k,k',s']$ within the $\Omega$ sector. The numerator is minimized by $k=\Phi^+$ and $s'=11$, whereas the denominator is larger for $k'=\Omega^+$ than for $k'=\Omega^-$ because $B>0$. Hence the unique minimum in this sector is
\begin{equation}
R[\Phi^+,\Omega^+,11]
=\frac{p_{k=\Phi^+}\bar\mu^2}{p_{k'=\Omega^+}}.
\end{equation}
The $s'=00$ trajectories have a larger local factor $(1+\mu)^2/4$ and therefore cannot yield a smaller value of $R[k,k',s']$.

In the $\Psi$ sector, the allowed final local labels $s'=01$ and $s'=10$ yield degenerate support values because $p_{s'_1}p_{s'_2}=\bar\mu(1+\mu)/4$. For each fixed $k'=\Psi^\pm$, the numerator is again minimized by $k=\Phi^+$. The only competing support values are therefore
\begin{subequations}
\begin{align}
R[\Phi^+,\Psi^+,01]
&=R[\Phi^+,\Psi^+,10]
=\frac{p_{k=\Phi^+}\bar\mu(1+\mu)}{p_{k'=\Psi^+}},\\
R[\Phi^+,\Psi^-,01]
&=R[\Phi^+,\Psi^-,10]
=\frac{p_{k=\Phi^+}\bar\mu(1+\mu)}{p_{k'=\Psi^-}}.
\end{align}
\end{subequations}
Thus, $R[\Phi^+,\Omega^+,11]$ is the unique minimum if and only if it is smaller than both distinct $\Psi$-sector support values.

Because $R[\Phi^+,\Psi^+,01]=R[\Phi^+,\Psi^+,10]$, it is sufficient to compare $R[\Phi^+,\Omega^+,11]$ with either member of this degenerate pair. For $s'=01$, the inequality $R[\Phi^+,\Omega^+,11]<R[\Phi^+,\Psi^+,01]$ gives
\begin{equation}
(1+\mu)p_{k'=\Omega^+}
>\bar\mu p_{k'=\Psi^+},
\end{equation}
which is equivalent to 
\begin{equation}
2\mu(1+\mu) +(1+\mu)B-\bar\mu^2\delta>0.
\end{equation}
This positivity follows from $B>\bar\mu\delta>\bar\mu^2\delta$. Therefore, $R[\Phi^+,\Omega^+,11]$ is always smaller than both degenerate $\Psi^+$ support values.

Because $R[\Phi^+,\Psi^-,01]=R[\Phi^+,\Psi^-,10]$, the remaining comparison can likewise be made with $s'=01$. The inequality $R[\Phi^+,\Omega^+,11]<R[\Phi^+,\Psi^-,01]$ gives
\begin{equation}
(1+\mu)p_{k'=\Omega^+}>\bar\mu p_{k'=\Psi^-}. 
\end{equation}
This inequality holds if and only if $G_\mu(r,\delta)>0$, where
\begin{equation}
G_\mu(r,\delta) = (1+\mu)B+\bar\mu^2\delta +2\mu(1+\mu)-4\bar\mu^2r. 
\end{equation}
Combining the $\Omega$, $\Psi^+$, and $\Psi^-$ comparisons proves the lemma.
\end{proof}

Using Lemma~\ref{lem:structured-minimum-condition}, we next prove Theorem~\ref{thm:structured-slice-violation}, which establishes that Horv\'ath's criterion is violated for the initial-state parameters in Eq.~\eqref{eq:structured-slice}. 

\begin{proof}[Proof of Theorem~\ref{thm:structured-slice-violation}]
For \(v=-\delta<0\), write
\(|\Omega^+\rangle=\omega^+_{00}|00\rangle+\omega^+_{11}|11\rangle\), where 
\begin{subequations}
\begin{align}
\omega^+_{00}&=\frac{2\mu+B}{\sqrt{\bar\mu^2\delta^2+(2\mu+B)^2}},\\
\omega^+_{11}&=-\frac{\bar\mu\delta}{\sqrt{\bar\mu^2\delta^2+(2\mu+B)^2}}.
\end{align}
\end{subequations}
Then
\begin{subequations}
\begin{align}
q[\Phi^+,11,\Omega^+,11]
&=\frac12\left(\bar\mu\omega^+_{00}\omega^+_{11}+\bar\mu^2\bigl(\omega^+_{11}\bigr)^2\right)\\
&=\frac{\bar\mu^2\delta(\bar\mu^2\delta-2\mu-B)}
       {2\left(\bar\mu^2\delta^2+(2\mu+B)^2\right)}.
\end{align}
\end{subequations}
Moreover,
\begin{equation}
B^2-\bar\mu^4\delta^2
=\bar\mu^2\delta^2\bigl(1-\bar\mu^2\bigr)+4\mu^2>0,
\end{equation}
so that \(B>\bar\mu^2\delta\). Multiplying this result by
\(p_{k=\Phi^+}=(1-r-\delta)/4\) gives
Eq.~\eqref{eq:structured-Qstar}. Since \(1-r-\delta>0\), the denominator is positive, and the factor \(\bar\mu^2\delta-2\mu-B\) is strictly negative. Hence \(\tilde{\mathcal Q}_m<0\). Lemma~\ref{lem:structured-minimum-condition} shows that \(G_\mu(r,\delta)>0\) is exactly the condition under which \(R[\Phi^+,\Omega^+,11]\) corresponds to the unique smallest support value. Corollary~\ref{Q_m} then implies the violation of Horv\'ath's criterion via the endpoint test. 
\end{proof}

\bibliography{ref}

\end{document}